\documentclass[12pt,a4paper,final]{iopart}

\usepackage{iopams}  
\usepackage{graphicx}
\usepackage[breaklinks=true,colorlinks=true,linkcolor=blue,urlcolor=blue,citecolor=blue]{hyperref}

\usepackage[square,sort,comma,numbers]{natbib}
\usepackage{tipa,comment}
\usepackage{tikz}
\usepackage{ragged2e}

\expandafter\let\csname equation*\endcsname\relax
\expandafter\let\csname endequation*\endcsname\relax
\usepackage{amssymb, amsthm, amsfonts, amsmath}
\usepackage{color,comment}
\usepackage{mathtools}
\usepackage{enumerate}
\usepackage{mathrsfs}
\usepackage{algorithm,algorithmic}
\usepackage{bbm}

\newtheorem{thm}{Theorem}[section]
\newtheorem{cor}[thm]{Corollary}
\newtheorem{prop}[thm]{Proposition}
\newtheorem{lem}[thm]{Lemma}

\theoremstyle{definition}

\theoremstyle{remark}
\newtheorem{rem}[thm]{Remark}

\makeatletter
\DeclarePairedDelimiter{\ceil}{\lceil}{\rceil}
\DeclarePairedDelimiter\abs{\lvert}{\rvert}
\DeclareFontFamily{U}{tipa}{}
\DeclareFontShape{U}{tipa}{m}{n}{<->tipa10}{}
\newcommand{\norm}[1]{\lVert#1\rVert}
\newcommand{\vect}[1]{\boldsymbol{#1}}

\makeatletter
\let\c@equation\c@thm
\makeatother
\numberwithin{equation}{section}

\begin{document}
	
	\title[Phase Retrieval: Regularization, Box Relaxation and Uniqueness] {One-Dimensional Phase Retrieval: Regularization, Box Relaxation and Uniqueness}
	
	\author{Wing Hong Wong$^1$, Yifei Lou$^2$, Stefano Marchesini$^3$, Tieyong Zeng$^1$}
	\address{$^1$Department of Mathematics, the Chinese University of Hong Kong, Hong Kong}
	\address{$^2$Mathematical Sciences Department, University of Texas Dallas, Richardson, TX 75080, USA}
	\address{$^3$Computational Research Division, Lawrence Berkeley National Laboratory, USA}
	\ead{
		\mailto{whwong@math.cuhk.edu.hk},
		\mailto{yifei.lou@utdallas.edu},
		\mailto{smarchesini@lbl.gov},
		\mailto{zeng@math.cuhk.edu.hk}}
	
	\begin{abstract}
		Recovering a signal from its Fourier magnitude is referred to as phase retrieval, which occurs in different fields of engineering and applied physics. This paper gives a new characterization of the phase retrieval problem. Particularly useful is the analysis revealing that the common gradient-based regularization does not restrict the set of solutions to a smaller set. 
		Specifically focusing on binary signals, we show that a box relaxation is equivalent to the binary constraint for Fourier-types of phase retrieval. We further prove that binary signals can be recovered uniquely up to trivial ambiguities under certain conditions. Finally, we use the characterization theorem to develop an efficient denoising algorithm.
	\end{abstract}
	
	\vspace{2pc}
	
	\submitto{Inverse Problems}

	\section{Introduction}
	In many fields of physics and engineering, one can only measure the magnitude of the Fourier Transform of 
	a discrete signal $\vect{x} \in \mathbb{C}^N$. Denote the discrete Fourier Transform by  $\mathcal{F}$. Recovering $\vect x$ from  $\abs{\mathcal{F}\vect{x}}$
	is referred to as \textit{phase retrieval} (PR), since the phase is completely lost in measurements. 
	Phase retrieval originated from X-ray crystallography \cite{Millane90,KimH91}, trying to determine  atomic and molecular structures of a crystal. This approach was later used to reconstruct an image of a sample with resolution at a nano-meter scale from its X-ray diffraction pattern, known as coherent diffraction imaging (CDI) \cite{MiaoCKS99}. 
	The PR techniques now occur in various applications such as astronomy \cite{Dainty1987} and laser optics \cite{SeifertSDLT06}; please refer to \cite{shechtman2014} for a contemporary overview.

	Phase retrieval is  a very challenging problem largely due to its nonconvexity and solutions being non-unique \cite{Hofstetter1964}. Specifically for the nonuniqueness (a.k.a., \textit{ambiguities}), there are trivial ambiguities and non-trivial ambiguities \cite{shechtman2014}.  Trivial ambiguities  of $\abs{\mathcal{F}\vect{y}} = \abs{\mathcal{F}\vect{x}}$ can be 
	summarized as follows,
	\begin{eqnarray} \label{eq:trivial}
	&&\mbox{global phase shift: } y_k = x_k \cdot e^{i\phi_0}\notag\\
	&&\mbox{conjugate inverse: } y_k = \overline{x_{-k}}\\
	&& \mbox{spatial shift: } y_k = x_{k+k_0}, \notag
	\end{eqnarray}
	where the indices are taken cyclically up to $N$, $\bar{\cdot}$ denotes the complex conjugate, and $\phi_0 \in [0, 2\pi), \ k_0 \in \mathbb{Z}$ are the phase shift and spatial shift, respectively. Note that every combination of \eqref{eq:trivial} is also a trivial ambiguity. 
	Non-trivial ambiguities of one-dimensional signals can be classified by the roots of the $Z$-transform of the \textit{autocorrelation} of the signal \cite{BeinertP15}, while almost all multi-dimensional signals only have non-trivial ambiguities \cite{BendoryBE2017},  since the $Z$-transform of their autocorrelation being reducible is of measure zero in the space of all polynomials \cite{BeinertP15, HayesM82}.

	For unique recovery of a real signal of size $N$ in up to trivial ambiguities, at least $2N-1$ random measurements are needed, provided the sampling matrix has full spark \cite{balan2006signal}. This result was later extended to the complex case in \cite{conca2015algebraic, FickusMNW14}, requiring at least $4N-4$ measurements. Other sufficient conditions for unique recovery include minimum phase signals \cite{HuangES16}, sparse signals with non-periodic support \cite{JaganathanOH13}, and signals with collision-free  \cite{RanieriCLV13}. For $s$-sparse signals in $\mathbb{R}^N$, the number of  Fourier magnitude measurements is in the order of $O(s\log(N/s))$ \cite{LiV13,ohlsson2014conditions}, while $\min \{ 2s, 2n-1 \}$ for random measurements \cite{WangX14}.
	
	In addition to taking more measurements than the ambient dimension, one often relies on regularization to refine the solution space with an attempt to reduce ambiguities. Stemming from image processing, a common choice  is a gradient-type formalism. 
	For example, Chang et al.~\cite{changLNZ16} considered the total variation, which is the $\ell_1$ norm of the gradient for phase retrieval. Computationally, many optimization algorithms can be used to solve the (regularized) phase retrieval problems, including alternating projections \cite{Gerchberg1972}, Wirtinger flow \cite{candes2015phaseIT},  alternating direction method of multipliers (ADMM) \cite{changLNZ16}, and a preconditioned proximal algorithm \cite{changMLZ18}.

	This paper contributes to a new set of characterization theorems for phase retrieval, indicating that gradient-based regularization is redundant to the magnitude measurements. 
	We also impose additional constraints on the underlying signal in order to resolve the ambiguities. In particular, we focus on binary signals \cite{KeiperKLP17} due to its simplicity and a wide variety of applications such as bar code  \cite{Esedoglu2003,lou2014partially} and obstacle detection \cite{LitmanLS98}.
	Specifically for phase retrieval, binary signals are considered in magnetism to describe the x-ray energies of some chemical compound films such as the SmCO$_5$ film \cite{ShiFNELSFTSMRK16}, and in block copolymers to describe films \cite{SteinWYDB16}. 
	It was observed empirically in \cite{hartCL18} that incorporating a box constraint into the ADMM framework, referred to ADMMB, often gives an exact recovery of binary signal, which motivates us to give a theoretical explanation. In this paper, we prove that the phase retrieval problem with binary constraint is equivalent to  phase retrieval with box relaxation. We describe a new type of \textit{trivial ambiguities} for binary phase retrieval and show that unique recovery is possible  under certain  conditions. A related work \cite{YuanWang18} proved binary signals that cannot be uniquely recovered by Fourier magnitude is a zero-measure set.
	Finally, we take the noise into consideration and develop a denoising algorithm.
	
	Our contributions are three-fold: (1) We give a characterization theorem (Theorem~\ref{thm: characterization}), revealing the fact that $\norm{\nabla^n \vect{x}}_2$ is completely determined by $\abs{\mathcal{F}\vect{x}}$ for an arbitrary integer order $n$. (2) We give thorough analysis of phase retrieval problem in a binary setting.
	We show that the box relaxation to binary constraint is equivalent to the original binary phase retrieval problem (Theorem~\ref{thm: ab box to binary}). We then describe a new type of ambiguities and  guarantee the uniqueness of binary phase under certain conditions. (3) We conduct a series of error analysis (Propositions~\ref{prop: denoise}--\ref{cor: denoising related to sparsity} and Corollary~\ref{cor: denoising for integral-valued signal}) of  phase retrieval, which motivates a new denoising scheme.
	
	The rest of the paper is organized as follows. In Section~\ref{sect:pre}, we set up notations and  review some practical ways of taking magnitude measurements. In Section~\ref{sect:char}, we give a new characterization theorem  and discuss its consequences. In Section~\ref{sect:box}, we prove that the phase retrieval of binary signals can be relaxed to the box constraint. Furthermore, we show it is possible to relax the set of vectors having the same norm to its convex hull.
	In Section~\ref{sect:ambi}, we describe a new type of ambiguities for binary signals and show that the unique recovery of binary signals is possible under some special circumstances. Several extensions from the  Fourier case   to other types of sampling schemes are presented in Section~\ref{sect:more_sampling}. In Section~\ref{sect:denoising}, we estimate recover accuracy with respect to noise  and propose a denoising algorithm that empirically yields better performance compared to a na\"ive approach. Section~\ref{sect:conclusion} concludes the paper. Appendix provides all the proofs for the theorems presented.

	\section{Preliminaries} \label{sect:pre}

	\subsection{Notations}
	Let $\vect{x}, \vect{y} \in \mathbb{C}^N$ be arbitrary signals, we define some notations that are used throughout the paper,
	\begin{itemize}
		\item $x_k$ denotes the $k$-th entry of $\vect{x}$, i.e. $\vect{x} = (x_0, x_1, x_2, \dots, x_{N-1})^T$
		\item $\norm{\vect{x}}_p$ denotes the $\ell_p$-norm of $\vect{x}$, i.e. $\norm{\vect{x}}_p = (\sum^{N-1}_{k=0} \abs{x_k}^p)^{\frac{1}{p}}$, where $p > 0$. For $p=0$, we define $\norm{\vect{x}}_0$ to be the $\ell_0$ ``norm'' by counting the number of its nonzero elements. 
		\item $\vect{e}_k$'s denotes the standard basis in $\mathbb{C}^N$, i.e. the vector with a $1$ in the $k$-th coordinate and $0$'s elsewhere, e.g., $\vect e_0 = (1, 0, 0, \dots, 0)^T$ and $\vect e_1 = (0, 1, 0, \dots, 0)^T$.
		\item $\mathcal{F}_{N \rightarrow M}: \mathbb{C}^N \to \mathbb{C}^M$ denotes the matrix representing  discrete Fourier transform (DFT), i.e.
		\begin{equation}\label{eq:oversampledFFT}         \mathcal{F}_{N\rightarrow M} = \begin{bmatrix}
		1 & 1 & 1 & \dots  & 1 \\
		1 & \omega & \omega^2 & \dots  & \omega^{N-1} \\
		1 & \omega^2 & \omega^4 & \dots  & \omega^{2(N-1)} \\
		\vdots & \vdots & \vdots & \ddots & \vdots \\
		1 & \omega^{M-1} & \omega^{2(M-1)} & \dots  & \omega^{(M-1)(N-1)} \\
		\end{bmatrix},
		\end{equation}
		where $\omega = e^{\frac{-2\pi i}{M}}$. Note that $\dfrac{1}{\sqrt{N}}\mathcal{F}_{N \rightarrow N}$ is unitary.  If $M>N$, we refer it  as an \textit{oversampling} Fourier matrix. 
		\item We define $$\vect{x} \odot \vect{y} = (x_0y_0, x_1y_1, \dots, x_{N-1}y_{N-1}),$$
		where $\odot$ denotes the Hadamard product (i.e. entrywise multiplication).
		\item The \textit{discrete (periodic) convolution} $\vect{x} * \vect{y}$ is defined by 
		\begin{equation}
		(\vect{x}*\vect{y})_j = \sum^{N-1}_{k=0} x_k y_{(j-k) \text{mod} N}, 
		\end{equation}
		for $j = 0, 1, \dots, N-1$.
		\item The \textit{(regular) autocorrelation} is defined by 
		\begin{equation}\label{eq:Aut}
		(\text{Aut}(\vect{x}))_j = \sum^{N-1}_{k=0} x_{(k+j)} \overline{x_k},
		\end{equation}
		where $j = -N+1, -N+2, \dots, N-1$ and $x_k = 0, \forall k < 0$ and $k > N-1$. 
		\item By replacing the zero boundary condition in the regular autocorrelation with periodic boundary condition,  we  consider \textit{periodic autocorrelation} defined as
		\begin{equation}\label{eq:Aut_p}
		(\text{Aut}_p (\vect{x}))_j = \sum^{N-1}_{k=0} x_{(k+j) \text{mod} N} \overline{x_k},  
		\end{equation}
		for $ j=0,1,\cdots, N-1.$ These definitions will be used in the proofs of some interesting results.
	\end{itemize}
	For the rest of the paper, we denote $\mathcal{F}_{N \rightarrow N}$ by $\mathcal{F}$, $\mathcal{F}_{N \rightarrow M}$ by $\mathcal{F}_{M}$, and omit $_{\text{mod} N}$  if the context is clear.
	
	\subsection{Sampling Schemes}
	
	In practice, there are numerous ways \cite{TrebinoDFSKRK97,Trebino00, JaganathanEH2016,candes2015phase,ChangELM18,BendoryEE18} to take  magnitude measurements of a signal. This paper develops new theoretical characterizations in PR focusing on the following sampling schemes.
	\begin{itemize}
		\item \textbf{Classic Fourier Transform.} One aims to find  an unknown signal $\vect{x} \in \mathbb{C}^N$ from the magnitude measurements $\vect{b} := \abs{\mathcal{F}{\vect{x}}}$, i.e. $$b_n = \abs[\Big]{\sum^{N-1}_{k=0} x_k e^{\frac{-2\pi k n i}{N}}}, \quad \forall n = 0, 1, \dots, N-1.$$ 
		\item \textbf{Oversampling Fourier Transform.}
		An $M$-point ($M > N$) oversampling discrete Fourier Transform (DFT) of a signal $\vect{x} \in \mathbb{C}^N$ is defined by $$b_n = \abs[\Big]{\sum^{N-1}_{k=0} x_k e^{\frac{-2\pi k n i}{M}}}, \quad \forall n = 0, 1, \dots, M-1.$$ One wants to recover an $N$-point signal $\vect x$ based on  $M$ measurements of $\abs{\mathcal{F}_M\vect{x}}$.  
		A typical choice of $M$ is  $M=2N$ \cite{Hayes1982}, which is experimentally adopted by Miao et al \cite{miao2000oversampling}. 
				Also, a sufficient number of measurements is crucial in avoiding false solutions  \cite{soldovieri2005global}.
			However, we show in Theorem~\ref{thm: characterization Oversampling} theoretically that more measurements (i.e. $M\geq N$) do not resolve ambiguities in the noiseless PR problem.

		\item \textbf{Short-Time Fourier Transform} (STFT) \cite{eldar2014sparse,JaganathanEH2016}. Let $\vect{x} \in \mathbb{C}^N$ be a signal of length $N$ and $\vect{w} \in \mathbb{C}^W$ be a window function of length $W$. The Short-Time Fourier Transform (STFT) of $\vect{x}$ with respect to $\vect{w}$ is defined as 
		\begin{equation} \label{eq:STFT}
		\vect{z}_{n,m} = \sum^{N-1}_{k=0} x_k w_{mL-k} e^{\frac{-2\pi k n i}{N}},    
		\end{equation}
		for $n = 0, 1, \dots, N-1$ and $m = 0, 1, \dots, R-1$, where $L$ denotes the separation in time between adjacent short-times sections, R = $\ceil{\frac{N+W-1}{L}}$ denotes the number of short-time sections considered, and $w_k := 0$ for all $k<0$ and $k>W-1$. 
		
		\item \textbf{Frequency-resolved optical gating trace} (FROG)
		\cite{TrebinoDFSKRK97,Trebino00,BendoryEE18}.
		Let $$z_{n,m} = x_n x_{n+mL},$$ where $L$ is a fixed integer. The FROG trace is equivalent to the one-dimensional Fourier magnitude of $z_{n,m}$ for each fixed $m$, i.e.
		\begin{equation}\label{eq:FROG}
		\abs{\hat{z}_{n,m}}^2 = \abs{\sum^{N-1}_{k=0}x_kx_{k+mL}  e^{\frac{-2\pi k n i}{N}}}^2,
		\end{equation}
		for $n = 0, \dots, N-1, m = 0, \dots, \ceil*{\frac{N}{L}}-1$. 
		
	\end{itemize}
	
	Both STFT and FROG make experimentally plausible means of additional phaseless measurements to improve the accuracy of phase retrieval. For example, the
	STFT measurements can be obtained by a set of shifted versions of a single mask, while FROG measures the product of the signal with a shifted version of itself. 
	It was claimed in \cite{eldar2014sparse} that the STFT magnitude leads to better performance than an oversampled DFT with the same number of measurements. 
	
	\section{Regularization and Constraint in Phase Retrieval} \label{sect:char}
	
	Mathematically, the Fourier-type of phase retrieval problems in one dimensional case is formulated as follows,
	\begin{equation*}
	\text{Find } \vect{x} \in \mathbb{C}^N, \text{ s.t. } \abs{\mathcal{F}_M\vect{x}} = \vect{b}.
	\end{equation*}
	It is desirable and often necessary to impose some regularization  term in order to regularize the solution and avoid ambiguities in PR as much as possible. A classic choice is the use of $\norm{\vect x}_2$ and $\norm{\nabla^n \vect{x}}_2$ to enforce the smoothness of an underlying signal $\vect x$, where $\nabla^n$ is the $n$-th order discrete finite difference operator.  For simple notations,  we define $\nabla^0 \vect x := \vect x$. In other words, a regularized PR problem can be expressed as
	\begin{equation*}
	\begin{aligned}
	\underset{\vect{x}}{\text{minimize}}
	\norm{\nabla^n \vect{x}}_2 
	\quad \text{s.t.}
	\quad \abs{\mathcal{F}\vect{x}} = \vect{b}.
	\end{aligned}
	\end{equation*}
	Unfortunately, Theorem~\ref{thm: same derivative norm} shows that $\norm{\nabla^n \vect{x}}_2$ is completely determined by $\abs{\mathcal{F}\vect{x}}$, which implies that such gradient-based regularization cannot resolve any ambiguities. But on the other hand, adding gradient-based regularizations may help to escape from local optima due to the nonconvex nature of the phase retrieval problem.

		\begin{thm} \label{thm: same derivative norm}
			Given $\vect{x}, \vect{y} \in \mathbb{C}^N$, if $\abs{\mathcal{F}\vect{x}} = \abs{\mathcal{F}\vect{y}}$, then ${\norm{\nabla^n\vect{x}}_2 = \norm{\nabla^n\vect{y}}_2},$ for all $n = 0, 1, 2, \dots$.
		\end{thm}

	One may  wonder whether it is helpful to take more measurements and then impose  regularizations. 
	Theorem~\ref{thm: same derivative norm oversampling} implies that  the gradient-type regularization  is  insufficient for the PR problem with more than phaseless $2N-1$ measurements.

		\begin{thm} \label{thm: same derivative norm oversampling}
			Let $M \geq 2N-1$, given $\vect{x}, \vect{y} \in \mathbb{C}^N$, if $\abs{\mathcal{F}_M\vect{x}} = \abs{\mathcal{F}_M\vect{y}}$, then ${\norm{\nabla^n\vect{x}}_2 = \norm{\nabla^n\vect{y}}_2},$ for all $n = 0, 1, 2, \dots$.
		\end{thm}

	\begin{rem}
		When $N < M < 2N-1$, gradient-based regularization may help. For example, let 
		$
		\vect x =(0,     0,     0,     0,     1,     0,     1,     0,     0,     1,     1)$ and 
		$\vect y =(0,     0,     0,     1,     0,     0,     0,     1,     0,     1,     1).$
		Both of them are of length $11$ and have the same $\abs{\mathcal{F}_{M}\vect{x}}=\abs{\mathcal{F}_{M}\vect{y}}$ for $M =12$, but  $\norm{\nabla^3\vect{x}}_2^2 = 7.5 \neq 7 = \norm{\nabla^3\vect{y}}_2^2$, where the third order finite scheme $\nabla^3\vect{x}$ is defined by  $(\nabla^3\vect{x})_k := -\frac{1}{2}x_{k-2} + x_{k-1} - x_{k+1} + \frac{1}{2}x_{k+2}$. 
	\end{rem}
	
	To prove Theorems~\ref{thm: same derivative norm}-\ref{thm: same derivative norm oversampling}, we need to review a classical result   that $\text{Aut}(\vect x)$ is determined by $\abs{\mathcal{F}_{2N-1} \vect{x}}$ and vice versa, as stated in Theorem~\ref{thm: characterization 2N-1}.
	
	\begin{thm}[\cite{BendoryBE2017, YuanWang18}] \label{thm: characterization 2N-1}
		Given $\vect{x}, \vect{y} \in \mathbb{C}^N$, the following statements are equivalent:
		\begin{enumerate}
			\item[(1)] $\abs{\mathcal{F}_{2N-1} \vect{x}} = \abs{\mathcal{F}_{2N-1} \vect{y}}$;
			\item[(2)] $\text{Aut}(\vect{x}) = \text{Aut}(\vect{y})$.
		\end{enumerate}
	\end{thm}
	
	We extend this analysis to  an arbitrary number of measurements (not just $2N-1$) as well as to period autocorrelation (from regular autocorrelation). Specifically in Theorem \ref{thm: characterization}, we  show that when $M=N$, $\text{Aut}_p(\vect{x})$ and $\norm{\vect{v} * \vect{x}}_2$ for $\vect{v} \in \mathbb{C}^N$ are determined by $\abs{\mathcal{F}\vect{x}}$, and vice versa. A similar result for $M\geq 2N-1$ is presented in Theorem~\ref{thm: characterization Oversampling}. 
	
	\begin{thm} \label{thm: characterization}
		Given $\vect{x}, \vect{y} \in \mathbb{C}^N$, the following statements are equivalent:
		\begin{enumerate}
			\item[(1)] $\abs{\mathcal{F}\vect{x}} = \abs{\mathcal{F}\vect{y}}$;
			\item[(2)] $\text{Aut}_p(\vect{x}) = \text{Aut}_p(\vect{y})$;
			\item[(3)] $\norm{\vect{v} * \vect{x}}_2 = \norm{\vect{v} * \vect{y}}_2$ $\forall \vect{v} \in \mathbb{C}^N$.
		\end{enumerate}
	\end{thm}
	
	\begin{thm} \label{thm: characterization Oversampling}
		Given $\vect{x}, \vect{y} \in \mathbb{C}^N$, $M \geq 2N-1$, the following statements are equivalent:
		\begin{enumerate}
			\item[(1)] $\abs{\mathcal{F}_M \vect{x}} = \abs{\mathcal{F}_M \vect{y}}$
			\item[(2)] $\text{Aut}(\vect{x}) = \text{Aut}(\vect{y})$
		\end{enumerate}
		Also, either (1) or (2) implies that $\text{Aut}_p(\vect{x}) = \text{Aut}_p(\vect{y})$ and $\norm{\vect{v} * \vect{x}}_2 = \norm{\vect{v} * \vect{y}}_2$ $\forall \vect{v} \in \mathbb{C}^N$. The converse does not necessarily hold.
	\end{thm}
	
	\begin{rem}
		For $M<2N-1$ and $M\neq N$, we cannot determine the autocorrelation from $M$ magnitude measurements of $|\mathcal F_M(\vect x)|$, due to an insufficient number of measurements.
	\end{rem}
	
	To the best of our knowledge, the equivalence of phaseless measurements to $\norm{\vect{v} * \vect{x}}_2, \forall \vect{v}$ is novel in the literature, which leads to useful consequences as characterized in Theorems~\ref{thm: same derivative norm} and  \ref{thm: same derivative norm oversampling}. In particular,
	Theorem~\ref{thm: same derivative norm} directly follows from Theorem~\ref{thm: characterization} (1) $\Rightarrow$ (3) and the fact that $\nabla^n\vect{x} = \vect{v}_n * \vect{x}$ for some $\vect{v}_n \in \mathbb{C}^N$. Similarly, Theorem~\ref{thm: same derivative norm oversampling} follows from Theorem~\ref{thm: characterization Oversampling}.

	\section{Box Relaxation to Binary Constraint} \label{sect:box}
	We now restrict our attention to binary signals $\vect{x} \in \{0, 1\}^N$, as another way of imposing additional prior knowledge to facilitate phase retrieval. 
	Mathematically, we formulate the binary phase retrieval problem as follows,
	\begin{equation*} \label{problem: 01 binary phase retrival}
	\text{Find } \vect{x} \in \{0,1\}^N, \text{ s.t. } \abs{\mathcal{F}\vect{x}} = \vect{b}. \tag{P}
	\end{equation*}
	Since the binary constraint is nonconvex, we relax it to a box constraint in a similar way as a linear problem \cite{mao2012reconstruction}:
	\begin{equation*} \label{problem: 01 box phase retrival}
	\text{Find } \vect{x} \in [0,1]^N, \text{ s.t. } \abs{\mathcal{F}\vect{x}} = \vect{b}. \tag{Q}
	\end{equation*}
	Clearly, if (\ref{problem: 01 binary phase retrival}) has a solution, then (\ref{problem: 01 box phase retrival}) also has a solution. The question is whether we can recover $\vect{x}$ from $\vect{b}$ through (\ref{problem: 01 box phase retrival}). Computationally, the binary constraint in  (P) can be posed as a minimization problem of $\vect x (1-\vect x)$  subject to $\vect x \in [0,1]^N$, which can be solved via the difference of the convex algorithm (DCA) \cite{TA,phamLe2005dc}. Each DCA iteration requires to a subproblem similar to the (Q) problem and it takes a few iterations for DCA to converge. Therefore, solving (Q) is computationally more efficient compared to (P).  
	Theoretically, we prove in Theorem~\ref{thm: ab box to binary} that all the solutions to (\ref{problem: 01 box phase retrival}) are solutions to (\ref{problem: 01 binary phase retrival}) and have the same number of $1$'s as the ground-truth signal.

	\begin{thm} \label{thm: ab box to binary}
		Given $0 \leq \alpha < \beta$, $\vect{x} \in \{\alpha, \beta \}^N$ and $\vect{y} \in [\alpha, \beta]^N$, if $\abs{\mathcal{F} \vect{x}} = \abs{\mathcal{F} \vect{y}}$, then $\vect{y} \in \{\alpha, \beta \}^N$ and $\vect y$ has the same number of $\alpha$'s and $\beta$'s as $\vect x$.
	\end{thm}
	If $\{0, 1\}^N$ in problem (\ref{problem: 01 binary phase retrival}) is replaced by a set such that every element has the same modulus, one can also relax the problem to its convex hull.
	
	\begin{thm} \label{thm: extreme point to convex hull}
		Suppose $\mathcal E$ is a set of complex number and there exists some constant $c > 0$ such that $\abs z= c \geq $ for all $z \in \mathcal E$. Given $\vect x \in \mathcal E^N$ and $\vect y \in \text{conv } \mathcal E^N$, if $\abs{\mathcal F \vect x} = \abs{\mathcal F \vect y}$, then $\vect y \in \mathcal E ^N$, where $\text{conv } \mathcal E$ denotes the convex hull of $\mathcal E$.
	\end{thm}
	
	We have a similar version of Theorem~\ref{thm: ab box to binary} when $\vect x \in \{-1,1\}^N$.
	\begin{cor} \label{cor: -11 box to binary}
		Given $\vect{x} \in \{ -1,1 \}^N$ and $\vect{y} \in [-1,1]^N$, if  $\abs{\mathcal{F} \vect{x}} = \abs{\mathcal{F} \vect{y}}$, then $\vect{y} \in \{ -1,1 \}^N$, and the number of $1$'s in $\vect{y}$ is the same as the number of $1$'s in $\vect{x}$ or the number of $-1$ in $\vect{x}$.
	\end{cor}

	We then characterize trivial ambiguities  for binary phase retrieval in  Section~\ref{sect:ambi} and extend to other sampling schemes in Section~\ref{sect:more_sampling}.
	
	\subsection{Ambiguities and Uniqueness} \label{sect:ambi}
	
	In addition to trivial  ambiguities  \eqref{eq:trivial} for general PR, there is another type of ambiguity in the binary setting. For example, one has $$\abs{\mathcal{F}(1, 1, 1, 1, 0, 0, 1, 0, 0, 0)^T} = \abs{\mathcal{F}(0, 0, 0, 0, 1, 1, 0, 1, 1, 1)^T},$$ in which the two signals are not related by \eqref{eq:trivial}, but rather by switching zeros and ones.   We present this  ambiguity for binary phase retrieval in Corollary~\ref{cor: flip same magnitude after Fourier Transform}. In fact, this result can be easily extended to the complex case:

		\begin{prop} \label{prop: more_ambi}
			Given $\vect x \in \mathbb C^N$, $\abs{\mathcal F \vect x} = \abs{\mathcal F (c\mathbbm 1 -\vect x)}$ if and only if $c = \frac{1+e^{-i\theta}} N\sum x_i$ for some $\theta \in [0,2\pi)$, where $\mathbbm{1}$ denotes the vector of all one's, i.e. $\mathbbm{1} = (1, 1, \dots, 1)^T$.
		\end{prop}
		
		Applying Proposition~\ref{prop: more_ambi} with $\theta = 0$ and noting that $\sum x_i = \norm{\vect x}_0$ for binary signal $\vect x$, one easily obtains:

	\begin{cor} \label{cor: flip same magnitude after Fourier Transform}
		Given $\vect{x} \in \{0,1\}^N$ and $N$ is even, if $\norm{\vect{x}}_0 = N/2$, then $\abs{\mathcal{F}\vect{x}} = \abs{\mathcal{F}(\mathbbm{1}-\vect{x})}$.
	\end{cor}

	As a by-product from the proof of Proposition \ref{prop: more_ambi}, we reveal an interesting fact, stating that if $\vect{x}$ and $\vect{y}$ have the same Fourier magnitude, then so do $(\mathbbm{1} - \vect{x})$ and $(\mathbbm{1} - \vect{y})$:
	
	\begin{prop} \label{prop: x to 1-x}
		Given $\vect{x}, \vect{y} \in \{0,1\}^N$, $\abs{\mathcal{F}\vect{x}} = \abs{\mathcal{F}\vect{y}}$ if and only if $\abs{\mathcal{F}(\mathbbm{1}-\vect{x})} = \abs{\mathcal{F}(\mathbbm{1}-\vect{y})}$.
	\end{prop}
	
	We show in Proposition~\ref{prop: 0123 uniqueness} that  the exact recovery of $\vect x$ up to trivial ambiguities \eqref{eq:trivial} is guaranteed when $\norm{\vect{x}}_0 \leq 3$ and $\norm{\vect{x}}_0 \geq N-3$. The proof uses the fact that $(\text{Aut}_p (\vect{x}))_k$ is the number of pairs of $1$'s with distance\footnote{Note that it is a wrap-around distance. For example, $x_0$ and $x_{N-1}$ are considered of distance $1$.} $k$ for a binary signal  $\vect{x}\in\{0, 1\}^N$. The combinatorial nature of $\text{Aut}_p (\vect{x})$ guarantees the uniqueness of $\vect{x}$ up to trivial ambiguities.

	\begin{prop} \label{prop: 0123 uniqueness}
		Given $\vect{x} \in \{0,1\}^N$, if $\norm{\vect{x}}_0 = 0, 1, 2, 3, N-3, N-2, N-1$ or $N$, then we can uniquely  recover $\vect{x}$ from $\abs{\mathcal{F}\vect{x}}$ up to the trivial ambiguities  \eqref{eq:trivial}.
	\end{prop}

	\begin{rem} \label{rem: ambiguities}
		The above does not hold for $4 \leq \norm{\vect{x}}_0 \leq N-4$ in general. For example,
		$(0,     0,     0,     0,     0,     1,     0,     1,     0,     0,     1,     1)^T$ and
		$(0,     0,     0,     0,     1,     0,     0,     0,     1,     0,     1,     1)^T$ have the same magnitude after Fourier Transform, but they are not related to each other by trivial ambiguities.
	\end{rem}
	
	Next, we would like to discuss the uniqueness in oversampling case. Recall that the $Z$-transform of a signal $\vect x \in \mathbb{C}^N$ is defined by $$P_{\vect{x}}(z) = \sum^{N-1}_{k = 0} x_k z^k,$$ which is a complex polynomial. The reciprocal polynomial $\tilde{P}_{\vect x}(z)$ of $P_{\vect x}(z)$ is defined by $\tilde{P}_{\vect x}(z) = z^n P_{\vect x}(z^{-1})$, where $n$ is the degree of the polynomial $P_{\vect x}(z)$. 
	If the $Z$-transform of an unknown binary signal $P_{\vect{x}}$ is either reciprocal or irreducible, then $\vect x$ can be recover uniquely up to conjugate inverse. Using this fact, the exact recovery up to trivial ambiguities in the oversampling case is characterized in Propositions~\ref{prop: uniqueness for equalling to conjugate inverse}-\ref{prop: probability of uniqueness}.
	
	\begin{prop} \label{prop: uniqueness for equalling to conjugate inverse}
		Given $M \geq 2N-1$ in the setting of the oversampling  Fourier PR, $\vect{x} \in \{ 0, 1\}^N$, if $x_n = x_{N-1-n}$ for all $n = 0, 1, \dots, N-1$, then we can recover $\vect{x}$ uniquely.
	\end{prop}
	
	\begin{prop} \label{prop: probability of uniqueness}
		Given $M \geq 2N-1$ in the setting of the oversampling  Fourier PR, we can recover a random unknown binary $\vect{x} \in \{ 0, 1\}^N$ uniquely up to the equivalence relation defined by $y_n = x_{N-1-n}$ with probability at least $\frac{c}{logN}$ for a constant $c > 0$. 
	\end{prop}

	Note that the factor $\frac{c}{logN}$ in Proposition~\ref{prop: probability of uniqueness} is  a lower bound. In fact, there is a conjecture in \cite{OdlyzkoP93} that most of all polynomial with $\{0, 1\}$ coefficients are irreducible. If it holds,  a much better lower bound can be expected.

	\subsection{Extensions to other sampling schemes}\label{sect:more_sampling}
	
	We  extend the analysis of Theorem~\ref{thm: ab box to binary} to the  oversampling case, STFT, and FROG  in Theorems~\ref{thm: 01 box to binary oversampling}--\ref{thm: -11 FROG}, respectively.  Also, it can be extended to $\{0, \alpha\}^N$, $\{-\alpha, \alpha\}^N$ simply by scaling, which are omitted.
	\begin{thm} \label{thm: 01 box to binary oversampling}
		Let $M \geq N$, given $\vect{x} \in \{0, 1\}^N, \vect{y} \in [0, 1]^N$, if $\abs{\mathcal{F}_{N \to M}\vect{x}} = \abs{\mathcal{F}_{N \to M}\vect{y}}$, then $\vect{y} \in \{ 0,1 \}^N$  and $\norm{\vect{y}}_0 = \norm{\vect{x}}_0$.
	\end{thm}

	\begin{thm} \label{thm: 01 STFT}
		Given $\vect{x} \in \{0,1\} ^N$ and $\vect{y} \in [0,1]^N$, if $\vect{x}$ and $\vect{y}$ have the same STFT under non-zero constant window, with $W \geq L$, as defined in (\ref{eq:STFT}), then $\vect{y} \in \{0,1\}^N$.
	\end{thm}

	\begin{thm} \label{thm: 01 FROG}
		Given $\vect{x} \in \{0,1\} ^N$ and $\vect{y} \in [0,1]^N$, if $\vect{x}$ and $\vect{y}$ have the same FROG trace \eqref{eq:FROG}, then $\vect{y} \in \{0,1\}^N$ and $\norm{\vect{y}}_0 = \norm{\vect{x}}_0$.
	\end{thm}

	\begin{thm} \label{thm: -11 FROG}
		Given $\vect{x} \in \{-1,1\} ^N$ and $\vect{y} \in [-1,1]^N$, if $\vect{x}$ and $\vect{y}$ have the same FROG trace, then $\vect{y} \in \{-1,1\}^N$.
	\end{thm}

	\begin{rem}
		Unlike Theorem~\ref{cor: -11 box to binary}, the number of $1$'s in $\vect x$ is not necessarily the same as the number of $1$'s nor $-1$'s in $\vect y$. For example, if we take $\vect x = (1, 1)^T$ and $\vect y = (1, -1)^T$, then $\vect x$ and $\vect y$ have the same FROG trace.
	\end{rem}
	
	\section{Denoising} \label{sect:denoising}
	The preceding sections focus on the noiseless case, where the measured data we obtain is $\vect{b} = \abs{\mathcal{F}\vect{x}}$. However, noise is inevitable in practice and there is a  need to develop denoising techniques for phase retrieval. For this purpose, we consider a corrupted measurement $\vect{\tilde{b}} = \vect{b} + \vect{\eta}$ with a noise term $\vect{\eta}$. In the proof of Theorem~\ref{thm: characterization} (specifically Lemma \ref{lemma: compute autocorrelation}), we reveal that $\mathcal{F}^{-1}(\vect{b} \odot \vect{b}) = \text{Aut}_p(\vect{x})$. If the noise $\vect{\eta}$ is small enough, then $\mathcal{F}^{-1}(\vect{\tilde{b}} \odot \vect{\tilde{b}})$ can be approximated by $\mathcal{F}^{-1}(\vect{b} \odot \vect{b})$, which is equivalent to $\text{Aut}_p(\vect{x})$. Proposition~\ref{prop: denoise} is about the  approximation error.
	
	\begin{prop} \label{prop: denoise}
		Given $\epsilon > 0$, $\vect{x} \in \mathbb{C}^N \setminus \{\vect{0}\}$, $\vect{b} = \abs{\mathcal{F}\vect{x}}$, $\tilde{\vect{b}} = \vect{b} + \vect{\eta}$ for some noise $\vect{\eta} \in \mathbb{C}^N$, if $\norm{\vect{\eta}}_{\infty} < \text{min} \{\dfrac{\epsilon}{4\norm{\vect{b}}_{\infty}}, \dfrac{\epsilon}{2}, 1\}$, then $\norm{\mathcal{F}^{-1}(\tilde{\vect{b}} \odot \tilde{\vect{b}}) - \text{Aut}_p(\vect{x})}_{\infty} < \epsilon$.
	\end{prop}
	Ideally, it would be helpful to analyze the error to the ground-truth signal, which is unfortunately impossible due to trivial and non-trivial ambiguities.

	In the following, we restrict the ground-truth signal $\vect{x} \in \{0, 1\}^N$ and observe a denoising scheme based on Proposition~\ref{prop: denoise} often gives good results. 
	For binary signal $\vect x$,
	we know $\text{Aut}_p(\vect{x}) \in \mathbb{Z}^N$. If the noise $\vect{\eta}$ is small such that $\norm{\mathcal{F}^{-1}(\tilde{\vect{b}} \odot \tilde{\vect{b}}) - \text{Aut}_p(\vect{x})}_{\infty} < 1/2$, we can  round off each entry of $\mathcal{F}^{-1}(\tilde{\vect{b}} \odot \tilde{\vect{b}})$ to the nearest integer to perform denoising.
	Since $(\text{Aut}_p (\vect{x}))_k$ is the number of pair of $1$'s with distance $k$,  $\norm{\mathcal{F}^{-1}(\tilde{\vect{b}} \odot \tilde{\vect{b}}) - \text{Aut}_p(\vect{x})}_{\infty} \geq 1/2$ means the measurements cannot give us the true number of pairs of $1$'s with distance $k$. In this circumstance, one should not expect to have a successful recovery. 

	\begin{prop} \label{cor: denoising related to sparsity}
		Given $\vect{x} \in \{0, 1\}^N \setminus \{ \vect 0 \}$, $\vect{b} = \abs{\mathcal{F}\vect{x}}$, $\tilde{\vect{b}} = \vect{b} + \vect{\eta}$ for some noise $\vect{\eta} \in \mathbb{C}^N$, if $\norm{\vect{\eta}}_{\infty} < \dfrac{1}{8\norm{\vect x}_0}$, then $\norm{\mathcal{F}^{-1}(\tilde{\vect{b}} \odot \tilde{\vect{b}}) - \text{Aut}_p(\vect{x})}_{\infty} < \dfrac{1}{2}$.
	\end{prop}

	Since $\norm{\vect x}_0\leq N,$ it is straightforward to have   Corollary~\ref{cor: denoising for integral-valued signal}. We can also express the error analysis in Proposition~\ref{cor: denoising related to sparsity}  in terms of signal-to-noise ratio (SNR). 
	
	\begin{cor} \label{cor: denoising for integral-valued signal}
		Given $\vect{x} \in \{0, 1\}^N$, $\vect{b} = \abs{\mathcal{F}\vect{x}}$, $\tilde{\vect{b}} = \vect{b} + \vect{\eta}$ for some noise $\vect{\eta} \in \mathbb{C}^N$, if $\norm{\vect{\eta}}_{\infty} < \dfrac{1}{8N}$, then $\norm{\mathcal{F}^{-1}(\tilde{\vect{b}} \odot \tilde{\vect{b}}) - \text{Aut}_p(\vect{x})}_{\infty} < \dfrac{1}{2}$.
	\end{cor}

	Recall SNR is defined by 
	$$\text{SNR}_\text{dB} = 10\log_{10} \dfrac{\norm{\vect x}_2^2}{\norm{\vect \eta}_2^2} .$$ 
	Proposition~\ref{cor: SNR} presents a condition  to safely round off each entry to  0 and 1.

	\begin{cor} \label{cor: SNR}
		Given $\vect{x} \in \{0, 1\}^N \setminus \vect 0$, if $$\text{SNR}_{\text{dB}} > 10\log_{10}(64) + 30\log_{10}\norm{\vect x}_0,$$ then $\norm{\mathcal{F}^{-1}(\tilde{\vect{b}} \odot \tilde{\vect{b}}) - \text{Aut}_p(\vect{x})}_{\infty} < \dfrac{1}{2}$.
	\end{cor}

	The proposed denoising scheme, referred to as \textit{rounding scheme}, is described as follows: given a corrupted measurement $\vect{\tilde{b}} \in \mathbb{C}^N$,
	\begin{enumerate}
		\item Round off each entry $\mathcal{F}^{-1}(\tilde{\vect{b}} \odot \tilde{\vect{b}})$ to nearest integer to get the autocorrelation $\text{Aut}_p(\vect{x})$. 
		\item Calculate $\vect{b} = \sqrt{\mathcal{F}(\text{Aut}_p(\vect{x}))}$, where the square root is taken entrywise.
		\item Solve the minimization problem:
		\begin{align}\label{eq:ADMMB}
		\vect{x^*}=\underset{\vect{x}}{\text{argmin}}   \norm{\abs{\mathcal{F}\vect{x}}-\vect{b}}^2_2
		\quad \text{s.t.}  \quad \vect{x} \in [0,1]^N.
		\end{align}
		
		\item Round off each entry of $\vect{x}^*$ to be either 0 or 1.
	\end{enumerate}
	
	We compare the proposed scheme with a \textit{na\"ive scheme} with the following steps: given a corrupted measurement $\vect{b} \in \mathbb{C}^N$, 
	\begin{enumerate}
		\item Solve the minimization problem \eqref{eq:ADMMB}.
		\item Round off each entry of $\vect{x}^*$ to be either 0 or 1.
	\end{enumerate}

		Both rounding and na\"ive schemes require to find a solution to \eqref{eq:ADMMB}, which can be solved via the alternating direction methods of multiplier (ADMM) \cite{boydPCPE11admm}.  We  summarize  in Algorithm 1  for Fourier phase retrieval subject to the $[0,1]$-box constraint \eqref{eq:ADMMB} via ADMM; for more details, please refer to  \cite{hartCL18}. 	
		Notice that ADMM requires two parameters:  $\rho_1$ and $\rho_2$.  We examine the effects of  these two parameters on the na\"ive scheme and the rouding scheme in terms of success rates. 
		We consider  a binary vector of length 50 with $5$ nonzero element as the ground-truth $\vect{x}_{\text{true}},$ which is  contaminated by noise with SNR$=16$ dB. 
		We choose $\rho_1, \rho_2$  among a candidate set of $\{10^{-6}, 10^{-5}, 10^{-4}, 10^{-3}, 10^{-2}\}$ and plot the success rates in Figure~\ref{fig: para_signal} based on 1000 random realizations; we declare a trial is successful if $\norm{\abs{\mathcal{F}\vect{x}_{\text{recovered}}} - \vect{b}} < 10^{-6}$. 
		We observe  no significant difference when $\rho_1 = \rho_2$ and hence we choose $\rho_1 = \rho_2 = 10^{-5}$ for both rounding and na\"ive schemes throughout the experiments. Figure~\ref{fig: para_signal} also shows that  our rounding scheme outperforms the na\"ive scheme when $\rho_1 = \rho_2$. 

		\begin{algorithm}[t]\label{algo:ADMMB}
			\begin{minipage}{20cm}
				\textbf{Input}: $\vect b$ and two positive parameters $\rho_1, \rho_2$\\
				\textbf{Initialize} $k = 0, \vect w^0 = 0, \vect d^0 = 0, \vect y^0 = 0, \vect z^0 = \vect b e^{i\vect \phi}$ with a random vector $\vect \phi$\\
				\indent 1
				\textbf{while} stopping conditions are not satisfied \textbf{do} \\
				\indent 2
				\quad $\vect x^{k+1} = (\rho_1+\rho_2)^{-1} (\rho_1 \mathcal F^* \vect z^k + \mathcal F^* \vect d^k + \rho_2 \vect y^k - \vect w^k)$\\
				\indent 3
				\quad $\vect y^{k+1} = \min(\max(\vect x^{k+1} + \vect w^k/ \rho_2,0),1)$ \\
				\indent 4
				\quad $\vect z^{k+1} = \textbf{prox}_{\rho_1} (\mathcal F \vect x^{k+1} - \vect d^k /\rho_1)$ \\
				\indent 5
				\quad $\vect d ^{k+1} = \vect d^k + \rho_1 (\vect z^{k+1} - \mathcal F \vect x^{k+1})$ \\
				\indent 6
				\quad $\vect w ^{k+1} = \vect w^k + \rho_2 (\vect x^{k+1} - \vect y^{k+1})$ \\
				\indent 7
				\quad k = k+1 \\
				\indent 8
				\textbf{end while} \\
				\textbf{Output} the solution $\vect x^* = \vect x ^k$
			\end{minipage}\\
			\caption{Fourier phase retrieval subject to a box constraint \eqref{eq:ADMMB} via ADMM.}
		\end{algorithm}

	\begin{figure}[t]  
		\includegraphics[width=.475\textwidth]{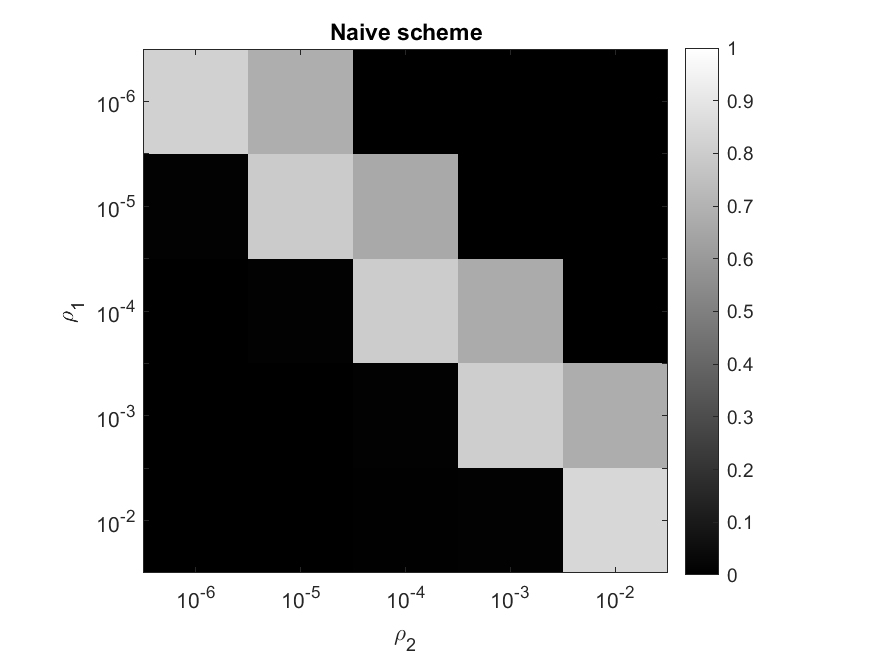}
		\includegraphics[width=.475\textwidth]{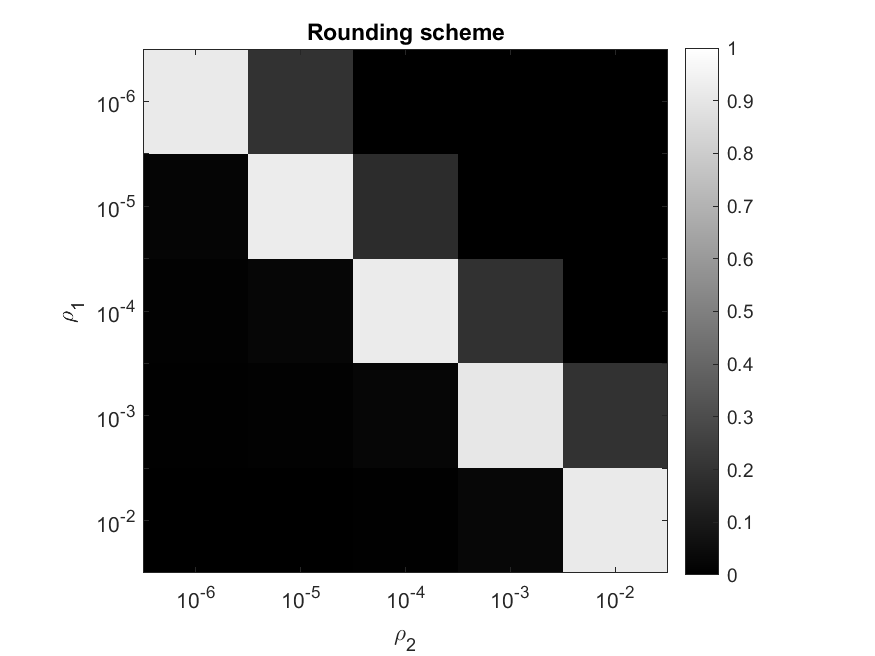}
		\caption{Influence of $\rho_1, \rho_2$ on the na\"ive scheme (left) and the rounding scheme (right) in terms of success rates when SNR $ = 16$ dB and $\norm{\vect x_{\text{true}}}_0 = 5$ for  $\vect{x}_{\text{true}}$  of length $50$.}
		\label{fig: para_signal}
	\end{figure}

	\subsection{Fourier phase retrieval}
	
	We then compare the performance of both schemes in terms of success rates. 
	We consider the ground-truth signal $\vect{x}_{\text{true}}$ is a binary vector  with different combinations of sparsity and noise levels in the Fourier measurements. 
		In particular, we examine ten sparsity levels ($1, 2, \dots, 10$) and generate the noisy measurements $\vect{\tilde{b}}$ by adding Gaussian noise with SNR $=(36 , 32, \dots, 0)$ dB.
		We plot the success rates of recovering signals of length 50 and 100 based on 1000 random realizations in Figure \ref{fig: 50_signal} and ~\ref{fig: 100_signal}, respectively.  Compared to the na\"ive scheme, the rounding scheme works much better when the signal is sparse, which is expected by Proposition~\ref{cor: denoising related to sparsity} that sparser signals allow for larger tolerance of the noise. According to Corollary~\ref{cor: SNR}, the exact recovery bound of SNR is calculated as $18+30\log_{10}\norm{\vect x}_0$, which aligns well with Figures~\ref{fig: 50_signal}-\ref{fig: 100_signal}. Figure~\ref{fig: false_reconstruction} gives some examples on false reconstructions, which implies that one scheme does not dominate the other, as there exist examples when  the na\"ive scheme succeeds and the rounding one fails, and vice versa. The conclusion that the rounding scheme is better is based on the success rates.  

	\begin{figure}[t]  
		\includegraphics[width=.475\textwidth]{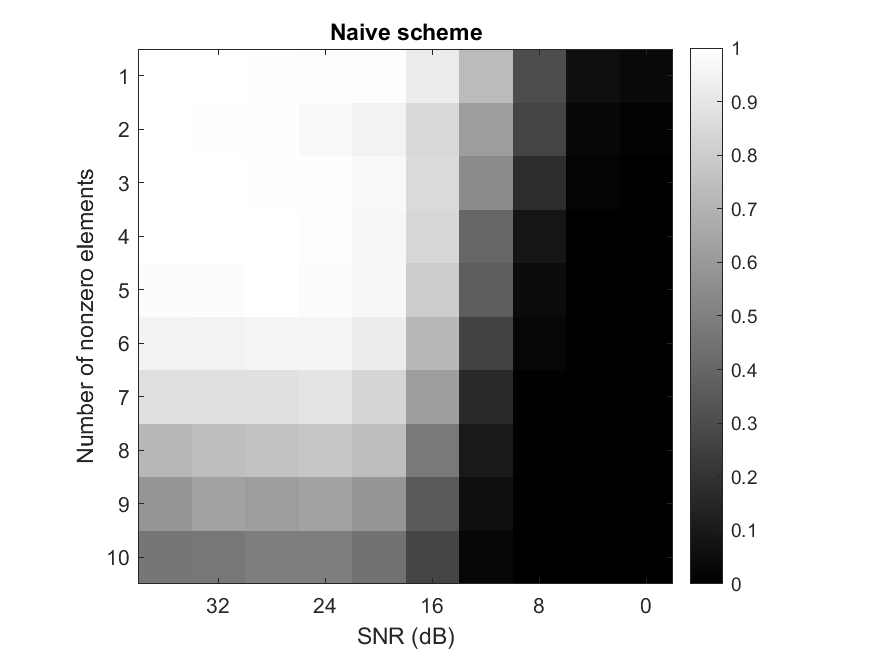}
		\includegraphics[width=.475\textwidth]{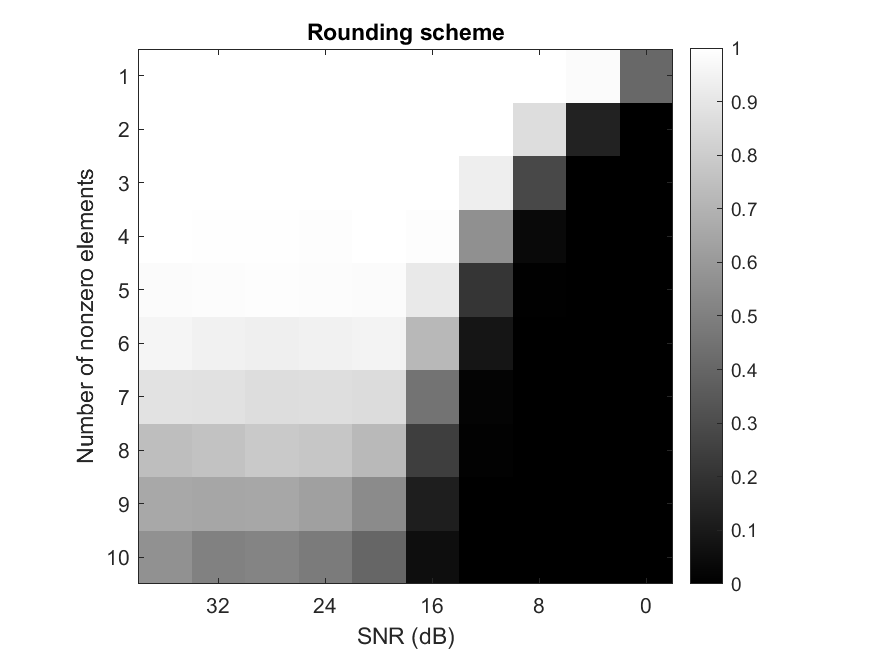}
		\caption{Comparison of the na\"ive scheme (left) and rounding scheme (right) in terms of success rates of Fourier phase retrieval for a signal of length  50. The value at each combination of sparsity and SNR is based on 1000 random realizations.}
		\label{fig: 50_signal}
	\end{figure}

	\begin{figure}[t]  
		\includegraphics[width=.475\textwidth]{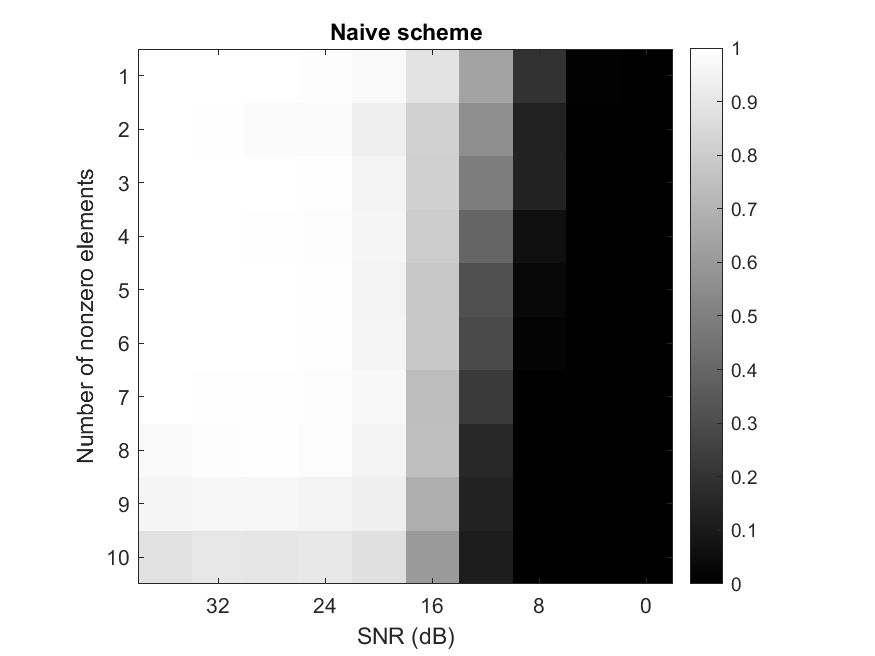}
		\includegraphics[width=.475\textwidth]{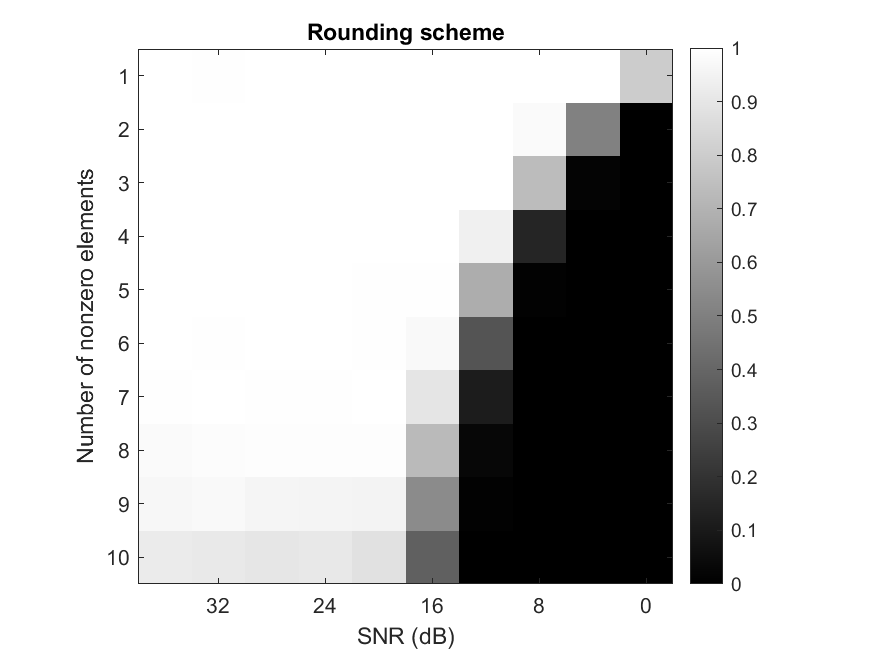}
		\caption{ Comparison of the na\"ive scheme (left) and rounding scheme (right) in terms of success rates of Fourier phase retrieval for a signal of length  100.} 
		\label{fig: 100_signal}
	\end{figure}

	\begin{figure}[t]  
		\includegraphics[width=.475\textwidth]{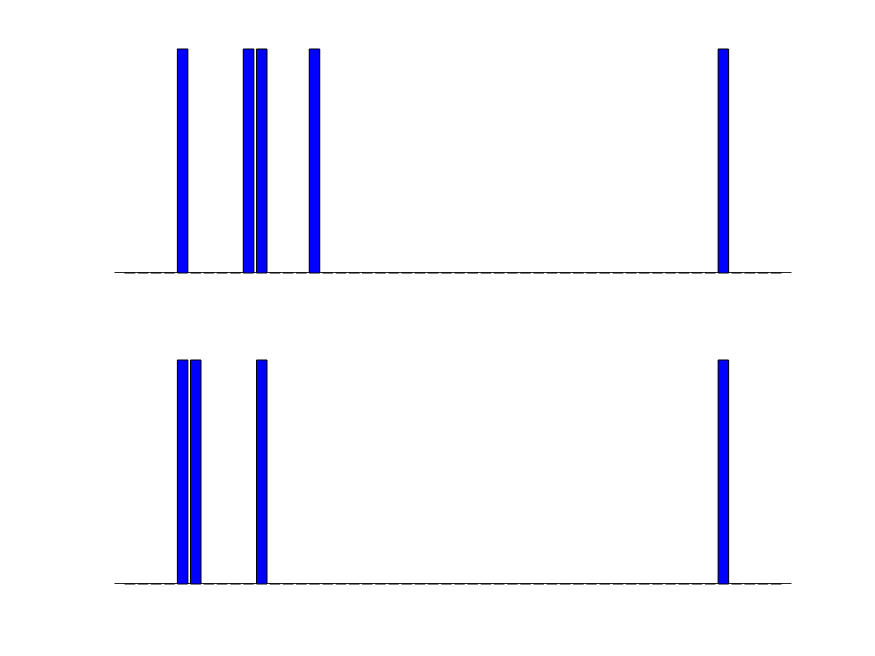}
		\includegraphics[width=.475\textwidth]{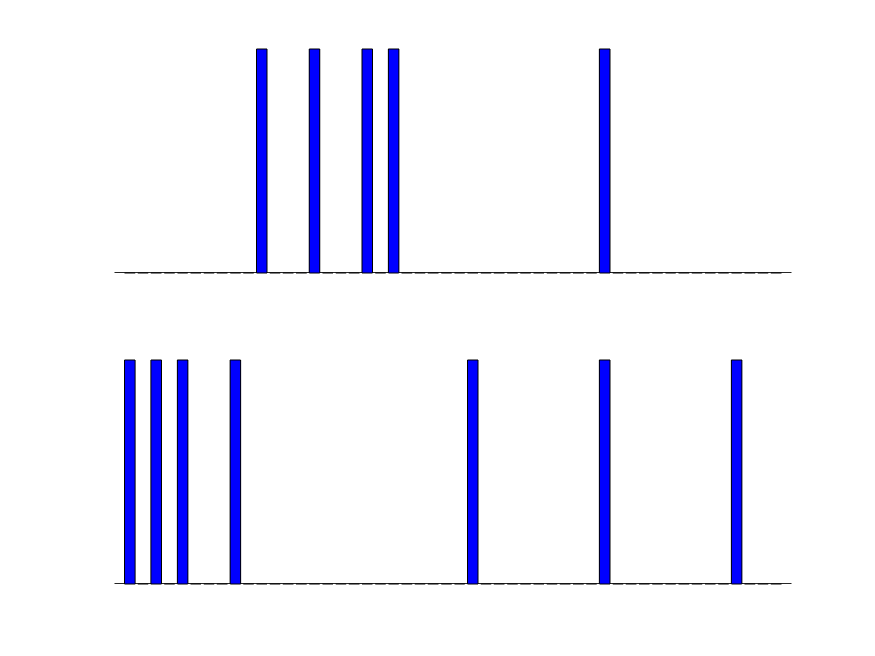}
		\caption{Failed reconstructions by the  na\"ive scheme (left) and the rounding scheme (right), when the other scheme succeeds. The ground-truth signals are plotted on the top, while the reconstructed ones are on the bottom. }
		\label{fig: false_reconstruction}
	\end{figure}

		\subsection{Extension to oversampling Fourier Transform}

		\begin{figure}[t]  
			\includegraphics[width=.475\textwidth]{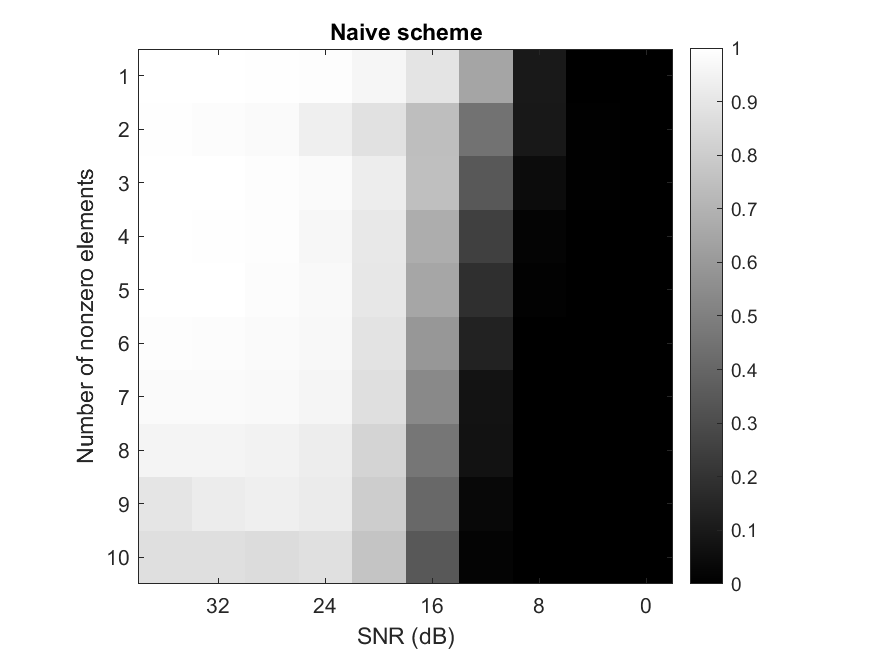}
			\includegraphics[width=.475\textwidth]{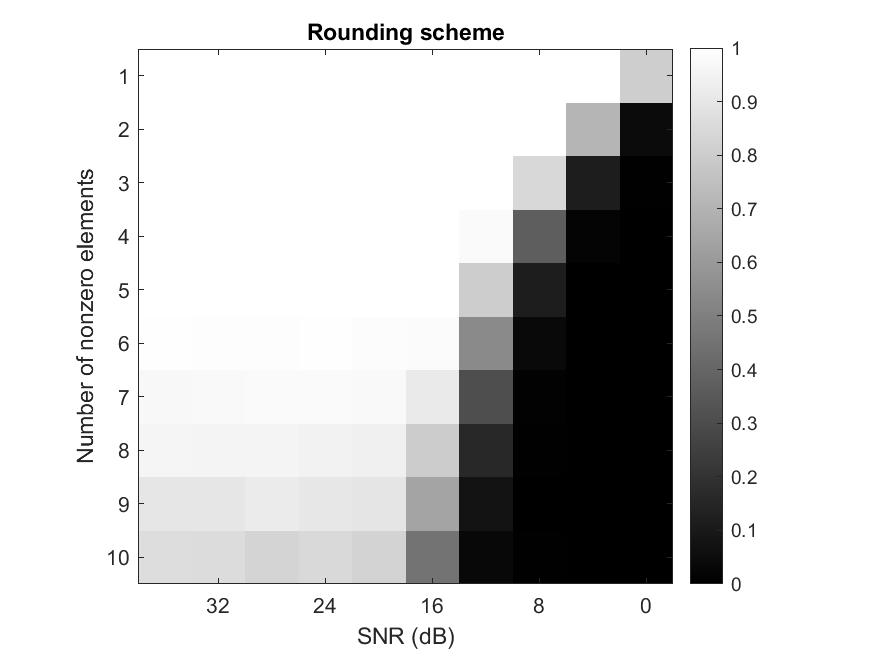}
			\caption{Comparison of the na\"ive scheme (left) and rounding scheme (right) in the case of oversampled Fourier phase retrieval.}
			\label{fig: over_signal}
		\end{figure}

		One may extend our method to oversampling schemes to find the periodic autocorrelation or the regular autocorrelation. 
		We conduct  numerical simulations for this case, while leaving the theoretical analysis for the future investigation.
		When the number of measurements do not match the number of coefficients in autocorrelation, one can perform a polynomial regression and round off to the nearest integer to find the autocorrelation, following equation \ref{eq: polynomial regression}. 
		The extension for the rounding scheme is summarized as follows, similar for the na\"ive scheme.
		
		\begin{enumerate}
			\item  Use polynomial regression to estimate the degree $2N-1$ polynomial $A(z)$ by $e^{\frac{2 \pi i k(N-1)}{M}} A(e^{-\frac{2 \pi i k}{M}}) = \tilde{ b}_k^2$.
			\item Round off each coefficient of $A(z)$ to the nearest integer to get the polynomial $B(z)$ with Aut$( \vect x)$ as its coefficient.
			\item Calculate $b_k = \sqrt{e^{\frac{2 \pi i k(N-1)}{M}} B(e^{-\frac{2 \pi i k}{M}})}$
			\item Solve the minimization problem:
			\begin{align}
			\vect{x^*}=\underset{\vect{x}}{\text{argmin}}   \norm{\abs{\mathcal{F}_M\vect{x}}-\vect{b}}^2_2
			\quad \text{s.t.}  \quad \vect{x} \in [0,1]^N.
			\end{align}
			
			\item Round off each entry of $\vect{x}^*$ to be either 0 or 1
		\end{enumerate}
		

		Again, we compare the performance of the na\"ive scheme and the rounding scheme in terms of success rates.
		We consider the ground-truth signal $\vect{x}_{\text{true}}$ is a binary vector of length 50 with different combinations of sparsity and noise levels. 
		We take $99$ oversampled Fourier magnitude measurements for each signal, i.e. $N=50, M=99$ in \eqref{eq:oversampledFFT}.
		We consider ten sparsity levels ($1, 2, \dots, 10$) and generate the noisy measurements $\vect{\tilde{b}}$ by adding Gaussian noise with SNR $=(36 , 32, \dots, 0)$ dB.
		In Figure~\ref{fig: over_signal}, we plot the success rates based on 1000 random realizations, which  
		shows the rounding scheme outperforms  the na\"ive one.

	\section{Conclusions} \label{sect:conclusion}
	In this paper, we improved upon an autocorrelation-based characterization of Fourier phase retrieval. 
	We discuss several choices of regularization terms and measurements. 
	Our analysis suggested that a gradient-based regularization, i.e. $\norm{\nabla^n \vect{x}}_2$, is redundant to the magnitude measurements, thus not  helpful to phase retrieval. 
	Furthermore, we proved that  binary signals can be recovered by imposing a box constraint. We also presented ambiguities and uniqueness  for binary phase retrieval. 
	Finally, we proposed a denoising scheme suggested by characterization theorems. 
	Since the proposed denoising scheme involves rounding, it is interesting to extend to  2D images, in which the measured data are often integer-valued. This will be our future work.  Another future direction involves theoretical analysis of oversampling schemes and noisy measurements for phase retrieval.
	
	\section*{Acknowledgements}
	WHW and TZ are sponsored in part by the National Natural Science Foundation of China under Grant 11671002, CUHK start-up and CUHK DAG 4053296, 4053342.
	YL is partially supported by NSF CAREER 1846690. SM was funded by the Center for Advanced Mathematics for Energy Research Applications, US Department of Energy, contract number DOE-DE-AC03-76SF00098.
	
	\appendix
	\section{Proof of Theorems \ref{thm: characterization} and \ref{thm: characterization Oversampling}}
	\renewcommand{\thethm}{A.\arabic{thm}}%
	
	To prove Theorem \ref{thm: characterization}, we introduce Lemma~\ref{lemma: compute autocorrelation} and \ref{lemma: spherical constraint}. Note that Lemma~\ref{lemma: compute autocorrelation} is a periodic version of a similar result in \cite[P. 215]{Beinert16} and Lemma~\ref{lemma: spherical constraint} is Parseval's Theorem.
	\begin{lem} \label{lemma: compute autocorrelation}
		$\mathcal{F}(\text{Aut}_p(\vect{x})) = \abs{\mathcal{F}\vect{x}} \odot \abs{\mathcal{F}\vect{x}}, \ \forall \vect{x} \in \mathbb{C}^N$.
	\end{lem}
	
	\begin{proof} It is straightforward that for all $j = 0, 1, \dots, N-1$, we have
		\begin{align*}
		&(\mathcal{F} (\text{Aut}_p(\vect{x})))_j
		= \sum^{N-1}_{m=0} \sum^{N-1}_{n=0} x_{n+m} \overline{x_n} \omega^{mj} \\
		=& \sum^{N-1}_{m=0} \sum^{N-1}_{n=0} x_m \overline{x_n} \omega^{(m-n)j} 
		= (\sum^{N-1}_{m=0} x_m \omega^{mj}) (\overline{\sum^{N-1}_{n=0}  x_{n} \omega^{nj}}) \\
		=& (\mathcal{F}\vect{x})_j \overline{(\mathcal{F}\vect{x})}_j 
		= \abs{(\mathcal{F}\vect{x})_j}^2.
		\end{align*}
	\end{proof}
	
	\begin{lem}[Application of Parseval's Theorem] \label{lemma: spherical constraint}
		Given $\vect{x}, \vect{y} \in \mathbb{C}^N$, if $\abs{\mathcal{F} \vect{x}} = \abs{\mathcal{F} \vect{y}}$, then $\norm{\vect{x}}_2= \norm{\vect{y}}_2$.
	\end{lem}
	\begin{proof}
		Since $\dfrac{1}{\sqrt{N}}\mathcal{F}$ is unitary, we have $$ \norm{\vect{y}}_2 =  \norm{\dfrac{1}{\sqrt{N}}\mathcal{F}\vect{y}}_2 = \dfrac{1}{\sqrt{N}}\norm{\abs{\mathcal{F}\vect{y}}}_2 = \dfrac{1}{\sqrt{N}}\norm{\abs{\mathcal{F}\vect{x}}}_2 = \norm{\vect{x}}_2.$$
	\end{proof}

	\begin{proof}[Proof of Theorem \ref{thm: characterization}]
		
		(1) $\Rightarrow$ (2). Suppose $\abs{\mathcal{F}\vect{x}} = \abs{\mathcal{F}\vect{y}}$,
		by Lemma \ref{lemma: compute autocorrelation}, $\mathcal{F}(\text{Aut}_p(\vect{x})) = \abs{\mathcal{F}\vect{x}} \odot \abs{\mathcal{F}\vect{x}}$. Hence, $\text{Aut}_p(\vect{x}) = \mathcal{F}^{-1} (\abs{\mathcal{F}\vect{x}} \odot \abs{\mathcal{F}\vect{x}}) = \mathcal{F}^{-1} (\abs{\mathcal{F}\vect{y}} \odot \abs{\mathcal{F}\vect{y}}) = \text{Aut}_p(\vect{y})$.

		(2) $\Rightarrow$ (1). Suppose $\text{Aut}_p(\vect{x})$ = $\text{Aut}_p(\vect{y})$. By Lemma \ref{lemma: compute autocorrelation}, we have $\abs{\mathcal{F}\vect{x}} = \sqrt{\mathcal F (\text{Aut}_p(\vect{x}))} = \sqrt{\mathcal F (\text{Aut}_p(\vect{y}))} = \abs{\mathcal{F}\vect{y}}$, where the square root is taken entrywisely.
		
		(1) $\Rightarrow$ (3). By the Convolution Theorem, we have $\forall \vect{v} \in \mathbb{C}^N$ and $j = 0, 1 , \dots, N-1$, $$(\mathcal{F}(\vect{v} * \vect{x}))_j = (\mathcal{F}\vect{v})_j \times (\mathcal{F}\vect{x})_j,$$ thus leading to, $$\abs{(\mathcal{F}(\vect{v} * \vect{x}))_j} = \abs{(\mathcal{F}\vect{v})_j} \abs{(\mathcal{F}\vect{x})_j}.$$ Similar result holds for $\mathcal{F}(\vect{v}*\vect{y})$. Since  $\abs{\mathcal{F}\vect{x}} = \abs{\mathcal{F}\vect{y}}$ (by assumption), we have 
		$\abs{\mathcal{F}(\vect{v} * \vect{x})} = \abs{\mathcal{F}(\vect{v} * \vect{y})}$, which implies that $\norm{\vect{v} * \vect{x}}_2 = \norm{\vect{v} * \vect{y}}_2$ by Lemma~\ref{lemma: spherical constraint}. 
		
		(3) $\Rightarrow$ (1). Suppose $\norm{\vect{v} * \vect{x}}_2 = \norm{\vect{v} * \vect{y}}_2$ for all $\vect{v} \in \mathbb{C}^N$. Since $\mathcal{F}$ is invertible, we can choose $\vect{v}_k = \mathcal F^{-1} \vect e_k\in \mathbb{C}^N$. Then we have
		\begin{align*}
		&\norm{\vect{v}_k * \vect{x}}_2^2 = \norm{\dfrac{1}{\sqrt{N}}\mathcal{F}(\vect{v}_k*\vect{x})}_2^2\\
		= &\dfrac{1}{N} \sum^{N-1}_{j=0} \abs{(\mathcal{F}(\vect{v}_k * \vect{x}))_j}^2 
		= \dfrac{1}{N} \sum^{N-1}_{j=0} \abs{(\mathcal{F}\vect{v}_k)_j}^2 \abs{(\mathcal{F}\vect{x})_j}^2 \\
		=& \dfrac{1}{N} \sum^{N-1}_{j=0} \abs{(\vect{e}_k)_j}^2 \abs{(\mathcal{F}\vect{x})_j}^2 
		= \dfrac{1}{N}\abs{(\mathcal{F}\vect{x})_k}^2.
		\end{align*}
		Similarly, we have $\norm{\vect{v}_k * \vect{x}}_2^2 = \norm{\vect{v}_k * \vect{y}}_2^2 =  \dfrac{1}{N}\abs{(\mathcal{F}\vect{y})_k}^2$ and hence $\abs{\mathcal{F}\vect{x}} = \abs{\mathcal{F}\vect{y}}$. 
	\end{proof}
	
	\begin{proof} [Proof of Theorem~\ref{thm: characterization Oversampling}]
		(1) $\Rightarrow$ (2) Define 
		\begin{equation} \label{eq: A_x}
		A_{\vect x}(z) = z^{N-1}\sum^{N-1}_{n=-(N-1)} (\text{Aut}(\vect{x}))_n z^n,
		\end{equation} and similarly for $A_{\vect y}(z)$. Note that 
		\begin{align} \label{eq: polynomial regression}
		e^{\frac{2 \pi i k(N-1)}{M}} A_{\vect x}(e^{-\frac{2 \pi i k}{M}}) 
		= \abs{(\mathcal{F}_{M}\vect{x})_k}^2 
		= \abs{(\mathcal{F}_{M}\vect{y})_k}^2 
		= e^{\frac{2 \pi i k(N-1)}{M}} A_{\vect y}(e^{-\frac{2 \pi i k}{M}}),
		\end{align}
		for all $k = 0, 1, \dots, M-1$. Since $A_{\vect x}$ and $A_{\vect y}$ are polynomials of degree at most $2N-1$, their coefficients are determined by $\abs{\mathcal{F}_M\vect{x}} = \abs{\mathcal{F}_M\vect{y}}$, which is a system of $M$ linear equations with $M \geq 2N-1$. Thus, $\text{Aut}(\vect{x}) = \text{Aut}(\vect{y})$.
		
		(2) $\Rightarrow$ (1). Suppose $\text{Aut}(\vect{x}) = \text{Aut}(\vect{y})$. Then $A_{\vect x}(z)$ = $A_{\vect y}(z)$. Since $M \geq 2N-1$, we have
		\begin{align} 
		\abs{(\mathcal{F}_{M}\vect{x})_k}^2 
		= e^{\frac{2 \pi i k(N-1)}{M}} A_{\vect x}(e^{-\frac{2 \pi i k}{M}})
		= e^{\frac{2 \pi i k(N-1)}{M}} A_{\vect y}(e^{-\frac{2 \pi i k}{M}}) 
		= \abs{(\mathcal{F}_M\vect{y})_k}^2, \label{eq:Aut determine DFT}
		\end{align}
		for all $k = 0, 1, \dots, M-1$.
		
		It remains to prove that (1) implies $\text{Aut}_p(\vect{x}) = \text{Aut}_p(\vect{y})$ and $\norm{\vect{v} * \vect{x}}_2 = \norm{\vect{v} * \vect{y}}_2$ $\forall \vect{v} \in \mathbb{C}^N$. This directly follows from (1) $\Rightarrow$ (2) and (1) $\Rightarrow$ (3) in Theorem~\ref{thm: characterization} by considering $M = N$ in equation~(\ref{eq:Aut determine DFT}).
	\end{proof}
	
	\section{Proof of Theorem~\ref{thm: ab box to binary} to Corollary~\ref{cor: -11 box to binary}}
	\renewcommand{\thethm}{B.\arabic{thm}}%
	\renewcommand{\thefigure}{B.1}%

	We given a geometry interpretation to facilitate the proof of Theorem~\ref{thm: ab box to binary}. For $\alpha = 0$ and $\beta = 1,$ we have $\vect{y}\in[0, 1]^N$.
	Lemma \ref{lemma: spherical constraint} implies that $\vect{y}$ must lie on a sphere, while $\sum x_i = \sum y_i$ implies that $\vect{y}$ must lie on a plane.  Therefore, the solution $\vect{y}$ must be on the intersection of these three sets, as illustrated in Figure~\ref{fig: 01 binary to box}.
	
	\begin{figure} 
		\begin{center} 
			\includegraphics[width=.75\textwidth]{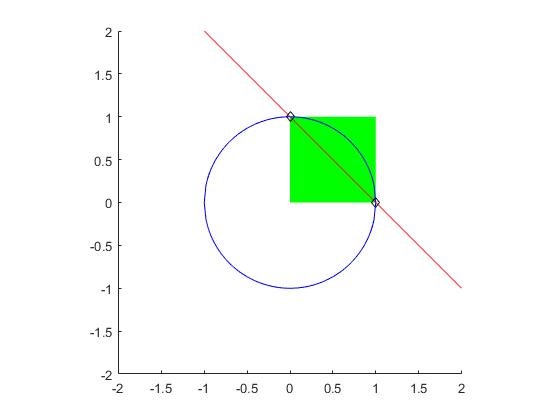}
			\caption{Illustration of  Theorem~\ref{thm: ab box to binary}  with $\alpha = 0$ and $\beta = 1$ when $\vect{x} = (1,0)$. The plane constraint is $y_1+y_2 = 1$, indicated by the red line. The sphere constraint is $y_1^2+y_2^2 = 1$, indicated by the blue circle. The box constraint is $0 \leq y_1, y_2 \leq 1$, indicated by the green square. The solutions to the  problem \eqref{problem: 01 box phase retrival} lie on the intersection of these $3$ sets, indicated by the black points. Both black points are ambiguous due to Proposition~\ref{cor: flip same magnitude after Fourier Transform}.}
			\label{fig: 01 binary to box}
		\end{center}
	\end{figure}

	\begin{proof} [Proof of Theorem~\ref{thm: ab box to binary}] 
		Rewrite $\abs{(\mathcal F \vect x)_0} = \abs{(\mathcal F \vect y)_0}$, we know that $\vect y$ lies on the plane $P: \sum x_i = \sum y_i$, which is convex. The box constraint $\vect y \in [\alpha, \beta]^N$ is also convex. Therefore, we define $\mathcal{C}:= P \cap [\alpha, \beta]^N$, which is a convex compact set. By Krein-Milman Theorem \cite[Theorem 3.23]{Rudin91}, $C$ is the closure of the convex hull of its extreme points. 
		
		We claim that the set of extreme points $\mathcal{E} = \{\vect{z}_i\}_{i \in \mathcal{I}}$ is a subset of points in $\{\alpha, \beta\}^N$ with the same number of $\alpha$'s and $\beta$'s as $\vect x$. Given $\vect w$ be an extreme point of $\mathcal{C}$, assume that $\vect w$ does not belong to $\{\alpha, \beta\}^N$. Since $\sum w_i = \sum x_i$ and $\vect x \in \{\alpha, \beta\}^N$, there exists some $i <  j$ such that $w_i, w_j \neq \alpha$ and $\beta$ (otherwise, we will have $\vect w \in \{\alpha, \beta\}^N$). Choose small $\epsilon > 0$ such that $w_i, w_j > \alpha + \epsilon$ and $w_i, w_j <\beta -\epsilon$. Let $\vect{w}_1 = (w_0, w_1, \dots, w_i + \epsilon, \dots, w_j -\epsilon, \dots, w_{N-1})^T$ and $\vect{w}_2 = (w_0, w_1, \dots, w_i - \epsilon, \dots, w_j +\epsilon, \dots, w_{N-1})^T$. Then $\vect{w}_1, \vect{w}_2 \in \mathcal{C}$ and $\vect w = \frac{1}{2}(\vect{w}_1 + \vect{w}_2)$, contradicting the fact that $\vect w$ is an extreme point of $\mathcal{C}$. Hence, we have $\vect w \in \{\alpha, \beta\}^N$. It follows from $\sum w_i = \sum x_i$  that $\vect w$ has the same number of $\alpha$'s and $\beta$'s as $\vect x$. Since $\mathcal{E}$ is a finite set, the convex hull of $\mathcal{E}$ is compact and thus equal to $\mathcal{C}$.

		Since $\vect y \in \mathcal C$, we write $\vect y = \sum \lambda_i \vect{z}_i$ for some $0 \leq \lambda_i \leq 1, \sum \lambda_i = 1$. Since $\vect{z}_i$ has the same number of $\alpha$'s and $\beta$'s as $\vect x$, then $f(\vect x) = f(\vect{z}_i)$ for all $i \in \mathcal{I}$, where $f(\vect w) := \norm{\vect w}_2^2$, which is a strictly convex function. By Lemma~\ref{lemma: spherical constraint}, we have $f(\vect y) = f(\vect x)$. If $\vect y$ does not belong to $\mathcal{E}$, then we have  
		$$f(\vect y) < \sum \lambda_i f(\vect{z}_i) = \sum \lambda_i f(\vect x) = f(\vect x) \sum \lambda_i  = f(\vect x),$$
		which is a contradiction. So  $\vect y \in \mathcal{E}$, i.e. $\vect y \in \{\alpha, \beta\}^N$ and has the same number of $\alpha$'s and $\beta$'s as $\vect x$.
	\end{proof}


	
	
	The proof of Theorem~\ref{thm: extreme point to convex hull} is based on the convexity to show that every $\vect y$ lies in the convex hull with some entries $y_i$ having smaller value than $\alpha$.
	
	\begin{proof} [Proof of Theorem~\ref{thm: extreme point to convex hull}]
		Since $\vect x \in \mathcal E ^N$, we have $\norm{\vect x}^2_2 = Nc^2.$ By Lemma \ref{lemma: spherical constraint}, we have $\norm{\vect y}^2_2 = \norm{\vect x}^2_2 = Nc^2$.
		
		Note that for all $i = 1, 2, \dots, N, y_i = \sum \lambda_{ik} z_{ik}$ for some $\sum \lambda_{ik} = 1, 0 \leq \lambda_{ik} \leq 1, z_{ik} \in \mathcal E$ and 
		$$ \abs{y_i} = \abs{\sum \lambda_{ik} z_{ik}} \leq \sum \lambda_{ik} \abs{z_{ik}} = \sum \lambda_{ik} c = c.$$
		
		If there exists some $y_i$ such that $y_i \in \text{conv } \mathcal E \setminus \mathcal E, $then $$ \abs{y_i} = \abs{\sum \lambda_{ik} z_{ik}} < \sum \lambda_{ik} \abs{z_{ik}} = \sum \lambda_{ik} c = c.$$ Now, $$ \norm{\vect y}^2_2 = \sum \abs{y_i}^2 < \sum c2 = Nc^2,$$ which leads to a contradiction. Thus, we must have $y_i \in \mathcal E$ for all $i = 1, 2, \dots, N$, i.e. $\vect y \in \mathcal E^N$
	\end{proof}
	
	\begin{proof} [Proof of Corollary~\ref{cor: -11 box to binary}]
		The fact that $\vect y \in \{-1, 1\}^N$ follows from Theorem~\ref{thm: extreme point to convex hull} directly. Now, $\abs{(\mathcal F \vect x)_0} = \abs{(\mathcal F \vect y)_0}$ implies that $\sum^{N-1}_{i = 0} x_i = \pm \sum^{N-1}_{i = 0} y_i$. Denote the number of $1$'s in $\vect{x}$ by $n_{\vect x}$, and define $n_{\vect y}$ similarly, then we have $n_{\vect x} - (N - n_{\vect x}) = \pm \big(n_{\vect y} - (N - n_{\vect y} \big)$. We either have $n_{\vect x} = n_{\vect y}$ or $n_{\vect x} = N$. The result now follows.
	\end{proof}

	\section{Proof of Propositions~\ref{prop: more_ambi}--\ref{prop: probability of uniqueness}}
	\renewcommand{\thethm}{C.\arabic{thm}}%

		\begin{proof} [Proof of Proposition~\ref{prop: more_ambi}]
			Suppose $\abs{\mathcal F \vect x} = \abs{\mathcal F (c\mathbbm 1 -\vect x)}$, i.e. $(\mathcal F \vect x)_0 = e^{i\theta}(\mathcal F (c\mathbbm 1 -\vect x))_0$ for some $\theta \in [0,2\pi)$. Thus, $\sum x_i = e^{i\theta}(Nc - \sum x_i)$, $c = \frac{1+e^{-i\theta}}N\sum x_i$.

			On the other hand, suppose  $c = \frac{1+e^{-i\theta}}N\sum x_i$ for some $\theta \in [0,2\pi)$. Since $\mathcal{F}\vect{x} + \mathcal{F}(c\mathbbm{1}-\vect{x}) = c\mathcal{F}\mathbbm{1} = Nc\vect{e}_0$, one has $(\mathcal{F}\vect{x})_j + (\mathcal{F}(c\mathbbm{1}-\vect{x}))_j = 0$ for  $j =1, 2, \dots, N-1$. 
			In particular, we obtain $\abs{(\mathcal{F}\vect{x})_j} = \abs{(\mathcal{F}(\mathbbm{1}-\vect{x}))_j}$ and clearly $\abs{(\mathcal F \vect x)_0} = \abs{(\mathcal F (c\mathbbm 1 -\vect x))_0}$  due to the choice of $c$.
		\end{proof}

	\begin{proof} [Proof of Proposition~\ref{prop: x to 1-x}]
		Similar to Proposition \ref{prop: more_ambi}, we have $$\abs{(\mathcal{F}(\mathbbm{1}-\vect{x}))_j} = \abs{(\mathcal{F}\vect{x})_j} = \abs{(\mathcal{F}\vect{y})_j} = \abs{(\mathcal{F}(\mathbbm{1}-\vect{y}))_j}$$ for $j =1, 2, \dots, N-1$.
		When $j=0$, we get $$(\mathcal{F}(\mathbbm{1}-\vect{x}))_0 = N - (\mathcal{F}(\vect{x}))_0 = N - (\mathcal{F}(\vect{y}))_0 =(\mathcal{F}(\mathbbm{1}-\vect{y}))_0.$$ 
		Therefore, $\abs{\mathcal{F}(\mathbbm{1}-\vect{x})} = \abs{\mathcal{F}(\mathbbm{1}-\vect{y})}$. Similar analysis for the other direction.
	\end{proof}

	\begin{proof} [Proof of Proposition~\ref{prop: 0123 uniqueness}] For a binary signal  $\vect{x}\in\{0, 1\}^N$, $(\text{Aut}_p (\vect{x}))_k$ is the number of pairs of $1$'s with distance $k$. As a result, when $\norm{\vect{x}}_0$ is either too small or too large, the uniqueness can be guaranteed thanks to the combinatorial nature of  $\text{Aut}_p (\vect{x}).$
		
		When $\norm{\vect{x}}_0 = 0$, $\vect{x}$ is the zero vector and hence the recovery is unique.
		
		When $\norm{\vect{x}}_0 = 1$, we get $\vect{x} = \vect{e}_k$ for some $k$, which is related by spatial shifts to each other. Therefore, the recovery is unique up to trivial ambiguities.
		
		When $\norm{\vect{x}}_0 = 2$,  we obtain the $\text{Aut}_p (\vect{x})$ from $\abs{\mathcal{F}\vect{x}}$ by Theorem \ref{thm: characterization}. Without loss of generality, up to spatial shift, we assume $x_0 = 1$. Let $k$ be the smallest positive number such that $(\text{Aut}_p (\vect{x}))_k$ is nonzero. Since $(\text{Aut}_p (\vect{x}))_k$ is equal to the number of pairs of $1$'s with distance $k$ and there are only two $1$'s in $\vect{x}$, i.e. only one pair of $1$'s. This pair must contain $x_0$. Say the pair contains $x_0$ and $x_j$. We know that $x_j$ and $x_0$ has distance $k$. Hence, $j$ = $k$ or $N-k$, i.e. we either have $x_0 = x_k = 1$ or $x_0 = x_{N-k} = 1$, which are spatial shifts of each other.
		
		When $\norm{\vect{x}}_0 = 3$, given $\abs{\mathcal{F}\vect{x}}$, we obtain  $\text{Aut}_p (\vect{x})$.  Let $k$ be the smallest positive number such that $(\text{Aut}_p (\vect{x}))_k$ is nonzero. Since there are three $1$'s in $\vect{x}$, there are $_3C_2$, i.e. $3$ pairs of $1$'s in $\vect{x}$. Thus, $(\text{Aut}_p (\vect{x}))_k = 1, 2$ or $3$. By spatial shift, we may assume one of the pairs contains $x_0$ and $x_k$. 
		
		If $(\text{Aut}_p (\vect{x}))_k = 2$ or $3$, then there is still at least one pair of $1$'s containing $x_0$ or $x_k$ and the remaining $1$. If it contains $x_0$, then the $1$ should lie in $x_{N-K}$ since $x_k$ is already occupied. If the pair contains $x_k$, by similar reasoning, the $1$ should lie in $x_{2k}$. In both cases, all three $1$'s are placed and these 2 cases are spatial shift of each other.
		
		If $(\text{Aut}_p (\vect{x}))_k = 1$, let $l$ be the smallest positive number greater than $k$ such that $(\text{Aut}_p (\vect{x}))_l$ is nonzero. By considering the position of $1$, we have 4 cases: $x_{N-l} = 1$, $x_{N-l+k} = 1$, $x_{l} = 1$ or $x_{l+k} = 1$. The cases that $x_{N-l+k} = 1$ and $x_l = 1$ are impossible, otherwise it will contradicts the minimality of $l$, $k$ and the fact that $(\text{Aut}_p (\vect{x}))_k = 1$, i.e. there is a pair of $1$ with distance $(l-k) < l$ while this pair is not the pair corresponding to the pair of distance $k$. Hence, we either have $x_0 = x_k = x_{l+k} = 1$ or $x_0 = x_k = x_{N-l} = 1$. Note that these two cases are equivalent to each other through conjugate inverse and spatial shift.
		
		The cases when $\norm{\vect{x}}_0 = N-3, N-2, N-1$ or $N$ now follow from above. If $\norm{\vect{x}}_0 = N-3, N-2, N-1$ or $N$, then $\norm{\mathbbm{1} - \vect{x}}_0 = 0, 1, 2$ or $3$. Hence, we can recover $(\mathbbm{1} - \vect{x})$ up to trivial ambiguities. Since $\vect{x} = \mathbbm{1} - (\mathbbm{1} - \vect{x})$, the recovery of $\vect{x}$ is unique up to trivial ambiguities.
	\end{proof}

	To prove Propositions~\ref{prop: uniqueness for equalling to conjugate inverse} and \ref{prop: probability of uniqueness}, we need to  introduce some  results in algebra. 
	Specifically, Theorem~\ref{thm: reciprocal} and \ref{thm: irreducible} are summarized from the proof  of \cite[Theorem 2.1]{YuanWang18}.

	\begin{thm}[Theorem 1 in \cite{KonyaginSV1999Otno}] \label{thm: 01 irreducible polynomial}
		Let $\vect x \in \{0,1\}^N$ with $x_0 = x_{N-1} = 1$, the $Z$-transform of $\vect x$ is irreducible with probability at least $c/\log N$ for some constant $c > 0$
	\end{thm}
	
	\begin{thm}\cite{filaseta2005irreducibility}\label{thm: reciprocal division}
		If a $f(x)$ is 0, 1 reciprocal polynomial and its constant term is $1$, then $f(x)$ is not divisible by a non-reciprocal polynomial in $\mathbb{Z} [x]$
	\end{thm}
	
	\begin{thm} \label{thm: reciprocal}
		Given $\vect x \in \{0,1\}^N$ and $\text{Aut}(\vect x)$, if $P_{\vect x} (z)$ is reciprocal, then there does not exist $\vect y \in \{0,1\}^N$  such that $\vect y \neq \vect x$ and $\text{Aut}(\vect y) = \text{Aut}(\vect x)$.
	\end{thm}
	
	\begin{proof}
		Define $A_{\vect x}$ by equation \ref{eq: A_x}. Then $A_{\vect x} = P_{\vect{x}}\tilde{P}_{\vect x}$. Write $P_{\vect x} =f_1f_2\dots f_k$ be its factorization such that each $f_j$ is irreducible. It follows from Theorem~\ref{thm: reciprocal division} that each $f_j$ is also reciprocal. 
		If there exists some $\vect y \in \{0,1\}^N$ such that $\text{Aut}(\vect y) = \text{Aut} (\vect x)$. Then $P_{\vect y} \tilde{P}_{\vect y} = A_{\vect y} = A_{\vect x} = P_{\vect x} \tilde{P}_{\vect x} = f_1^2 f_2^2 \dots f_k^2.$ If $f_j$ divides $P_{\vect y}$, we also have $\tilde{f_j} = f_j$ divide $\tilde{P}_{\vect y}$, and vice versa. So $(P_{\vect y}/f_1) (\tilde{P}_{\vect y}/f_1) =  f_2^2 f_3^2 \dots f_k^2$. Inductively, we have $P_{\vect y} = \tilde{P}_{\vect y} = f_1f_2\dots f_k = P_{\vect x}$, i.e. $\vect x = \vect y$.
	\end{proof}

	\begin{thm} \label{thm: irreducible}
		Given $\vect x \in \{0,1\}^N$, $\text{Aut}(\vect x)$, if $P_{\vect x} (z)$ is irreducible, then the only $\vect y \in \{0,1\}^N$ satisfying $\text{Aut}(\vect y) = \text{Aut}(\vect x)$ is either $\vect x$ or $\vect z$, which is defined by $z_n = x_{N-1-n}$. for $n = 0, 1, \dots, N-1$.
	\end{thm}
	
	\begin{proof}
		Similar to the above, we have $P_{\vect y} \tilde{P}_{\vect y} = P_{\vect x} \tilde{P}_{\vect x}$. Since $P_{\vect x}$ is irreducible, we have either $P_{\vect x}$ divides $P_{\vect y}$ or $\tilde{P}_{\vect y}$. 
		Suppose the first case. We have $\tilde{P}_{\vect x}$ divides $\tilde{P}_{\vect y}$. Together with the fact that $P_{\vect x}$ divides $P_{\vect y}$, this implies $P_{\vect x} = P_{\vect y}$ and $\tilde{P}_{\vect x} = \tilde{P}_{\vect y}$, i,e, $\vect y = \vect x$.
		Similarly, the second case implies that $P_{\vect y} = \tilde{P}_{\vect x}$, i,e, $\vect y$ is the conjugate inverse of $\vect x$.
		
	\end{proof}
	

	\begin{proof} [Proof of Proposition~\ref{prop: uniqueness for equalling to conjugate inverse}]
		By Theorem \ref{thm: characterization Oversampling}, $\text{Aut}(\vect{x})$ is uniquely determined when $M \geq 2N-1$. Note that $P_{\vect{x}}(z)$, the $Z$-transform of $\vect{x}$, is a reciprocal polynomial since $\vect{x}$ is equal to its conjugate inverse. According to Theorem~\ref{thm: reciprocal}, there does not exist $\vect y \neq \vect x$ such that $\text{Aut}(\vect{y}) = \text{Aut}(\vect{x})$. Therefore, we can uniquely recover $\vect{x}$  from $\text{Aut}(\vect{x})$ up to trivial ambiguities.
	\end{proof}
	
	\begin{proof} [Proof of Propsotion~\ref{prop: probability of uniqueness}]
		Theorem~\ref{thm: 01 irreducible polynomial} shows that for a random binary $\vect{x}$ with $x_0 = x_{N-1} = 1$,  the $Z$-transform of $\vect{x}$ ($P_{\vect{x}}(z)$) is irreducible with probability at least $\frac{c'}{logN}$ for a constant $c' > 0$. Note that we have $2^{N-2}$ binary signals under the constraint $x_0 = x_{N-1} = 1$ while we have $2^N$ binary signals in total. For a random binary $\vect{x}$, $P_{\vect{x}}(z)$ is irreducible with probability at least $\frac{1}{4} \frac{c'}{logN}$ = $\frac{c}{logN}$, where $c = \frac{c'}{4} > 0$ is a fixed constant. The remaining now follows from Theorem~\ref{thm: irreducible} directly.
	\end{proof}

	\section{Proof of Theorems~\ref{thm: 01 box to binary oversampling}--\ref{thm: -11 FROG}}
	
	\begin{proof} [Proof of Theorem~\ref{thm: 01 box to binary oversampling}]
		Write $\abs{\mathcal{F}_{N \to M}\vect{x}} = \abs{\mathcal{F}_{M \to M} \tilde{\vect{x}}}$ with $$\tilde{\vect{x}} = (x_0, x_1, \dots, x_{N-1}, 0, 0, \dots, 0)^T \in \{0, 1\}^M.$$ Similarly, we write $\abs{\mathcal{F}_{N \to M}\vect{y}} = \abs{\mathcal{F}_{M \to M} \tilde{\vect{y}}}$ and define $\tilde{\vect{y}}$. Note that $\tilde{\vect{x}} \in \{0, 1\}^M$, $\tilde{\vect{y}} \in [0, 1]^M$ and $\abs{\mathcal{F}_{M \to M}\tilde{\vect{x}}} = \abs{\mathcal{F}_{M \to M}\tilde{\vect{y}}}$. By Theorem \ref{thm: ab box to binary} with $\alpha = 0$ and $\beta = 1$, one has $\tilde{\vect{y}} \in \{0, 1\}^M$ and $\norm{\tilde{\vect{y}}}_0 = \norm{\tilde{\vect{x}}}_0$. Since $\tilde{\vect{x}}$ and $\tilde{\vect{y}}$ are obtained by appending zeros to $\vect{x}$ and $\vect{y}$, we have $\vect{y} \in \{0, 1\}^N$ and $\norm{\vect{y}}_0 = \norm{\tilde{\vect{y}}}_0 = \norm{\tilde{\vect{x}}}_0 = \norm{\vect{x}}_0$.
	\end{proof}
	
	\begin{proof}[Proof of Theorem~\ref{thm: 01 STFT}]
		Without loss of generality,  we may assume the windows $\vect w$ is an all one vector $\mathbbm{1}$ by scaling.
		Recall the STFT of $\vect x$ is defined by $$\vect{z}_{n,m} = \sum^{N-1}_{k=0} x_k w_{mL-k}  e^{\frac{-2\pi k n i}{N}}.$$
		
		Since $W \geq L$, for each $l = 0, 1, \dots, N-1$, there is some $m$ such that $w_{mL-l} = 1$. For such $m$, define $\tilde{x}_k = x_k w_{mL-k}$ for all $k = 0, 1, \dots, N-1$ and define $\tilde{\vect{y}}$ in a similar way. Then, $\tilde{\vect{x}} \in \{0,1\}^N$ and $\tilde{\vect{y}} \in [0, 1]^N$ by our assumption on $\vect w$.
		
		Now, $\abs{\mathcal{F}\tilde{\vect{x}}} = \vect{z}_{\cdot, m} = \abs{\mathcal{F}\tilde{\vect{y}}}$. Applying Theorem~\ref{thm: ab box to binary} with $\alpha = 0$ and $\beta = 1$, we have $\tilde{\vect{y}} \in \{0, 1\}^N$. In particular,
		$y_l = y_l w_{mL-l} = \tilde{y}_l \in  \{0, 1\}$. Since $l$ is arbitrary, we have $\vect y \in \{0, 1\}^N$.
	\end{proof}
	
	\begin{proof} [Proof of Theorem~\ref{thm: 01 FROG}]
		Denote $\abs{\hat{z}_{k,m}}^2$ and $\abs{\hat{w}_{k,m}}^2$ be the FROG trace \eqref{eq:FROG} of $\vect{x}$ and $\vect{y}$, respectively.
		We consider $m = 0$ and define $\vect{z}_0 = (z_{0,0}, z_{1,0}, \dots, z_{N-1,0})^T$ and similarly for $\vect{w}_0$. As $x_n \in \{0, 1\}$, we obtain $z_{n,0} = x_n^2 = x_n$ and  $w_{n,0} = y_n^2 \in [0,1]$. Now, our assumption translates to  $\abs{\hat{z}_{k,0}} = \abs{\hat{w}_{k,0}}$ for $k = 0, \dots, N-1$, i.e. $\abs{\mathcal{F}\vect{z}_0} = \abs{\mathcal{F}\vect{w}_0}$. Since $\vect{z}_0 \in \{0, 1\}^N$ and $\vect{w}_0 \in [0,1]^N$, we have $\vect{w}_0 \in \{0,1\}^N$ by Theorem \ref{thm: ab box to binary} with $\alpha = 0$ and $\beta = 1$, i.e. $w_{n,0} = y_n^2 \in \{0,1\}$ for all $n = 0, \dots, N-1$. Therefore, we obtain $\vect{y} \in \{0,1\}^N$, which implies that $\vect{y} = \vect{w}_0$.
		Since $\vect{x} = \vect{z}_0$, we have $\abs{\mathcal{F}\vect{x}} = \abs{\mathcal{F}\vect{z}_0} = \abs{\mathcal{F}\vect{w}_0} = \abs{\mathcal{F}\vect{y}}$ and $\norm{\vect{y}}_0 = \norm{\vect{x}}_0$ by Theorem \ref{thm: ab box to binary}.
	\end{proof}
	
	\begin{proof} [Proof of Theorem~\ref{thm: -11 FROG}]
		The proof is similar to the proof of Theorem~\ref{thm: 01 FROG} by noting that $z_{n,0} = x_n^2 = 1 \in \{-1,1\}$ and using Theorem~\ref{cor: -11 box to binary}.
	\end{proof}
	
	\section{Proof of Proposition~\ref{prop: denoise} to Corollary~\ref{cor: denoising for integral-valued signal}}
	\renewcommand{\thethm}{E.\arabic{thm}}%

	\begin{proof} [Proof of Proposition~\ref{prop: denoise}]
		\begin{align*}
		& \norm{\mathcal{F}^{-1}(\tilde{\vect{b}} \odot \tilde{\vect{b}}) - \text{Aut}_p(\vect{x})}_{\infty}  
		= \norm{\mathcal{F}^{-1}(\tilde{\vect{b}} \odot \tilde{\vect{b}}) - \mathcal{F}^{-1}(\vect{b} \odot \vect{b})}_{\infty} \\
		\leq & \norm{\mathcal{F}^{-1}(\tilde{\vect{b}} \odot \tilde{\vect{b}}) - \mathcal{F}^{-1}(\vect{b} \odot \vect{b})}_2 
		= \dfrac{1}{\sqrt{N}} \norm{\tilde{\vect{b}} \odot \tilde{\vect{b}} - \vect{b} \odot \vect{b}}_2 \\
		\leq & \norm{\tilde{\vect{b}} \odot \tilde{\vect{b}} - \vect{b} \odot \vect{b}}_{\infty} 
		= \norm{2 \vect{b} \odot \vect{\eta} + \vect{\eta} \odot \vect{\eta}}_{\infty} \\
		\leq & 2 \norm{\vect{b}}_{\infty} \norm{\vect{\eta}}_{\infty} + \norm{\vect{\eta}}_{\infty}^2 
		< \dfrac{\epsilon}{2} + \dfrac{\epsilon}{2}  \leq \epsilon, 
		\end{align*}
		where the first and second inequalities come from the fact that $\norm{\vect{x}}_{\infty} \leq \norm{\vect x}_2 \leq \sqrt N \norm{\vect{x}}_{\infty}$ for all $\vect x \in \mathbb{C}^N$ and the third inequality comes from the fact that $\norm{\vect{x} \odot \vect{y}}_{\infty} \leq \norm{\vect{x}}_{\infty} \norm{\vect{y}}_{\infty}$ for all $\vect{x}, \vect{y} \in \mathbb{C}^N$.
	\end{proof}

	\begin{proof} [Proof of Proposition~\ref{cor: denoising related to sparsity}]
		Note that $$b_n = \abs[\Big]{\sum^{N-1}_{k=0} x_k e^{\frac{-2\pi k n i}{N}}} \leq \sum^{N-1}_{k=0} \abs{x_k} = \norm{\vect x}_1 = \norm{\vect x}_0.$$ So $\norm{\vect b}_{\infty} \leq \norm{\vect x}_0.$
		
		Let $\epsilon = 1/2$, we have $\norm{\vect{\eta}}_{\infty} < \dfrac{1}{8\norm{\vect x}_ 0} \leq \dfrac{1}{8\norm{\vect b}_{\infty}} =  \dfrac{\epsilon}{4\norm{\vect b}_{\infty}}$. Also, since $\vect x \neq \vect 0$, $\norm{\vect x} \geq 1$. $\norm{\vect{\eta}}_{\infty} < \dfrac{1}{8\norm{\vect x}_ 0} \leq \dfrac{1}{8} \leq \min \{ \dfrac{\epsilon}{2} , 1\}.$ The remaining  follows from Proposition \ref{prop: denoise}.
	\end{proof}

	\begin{proof} [Proof of Corollary~\ref{cor: denoising for integral-valued signal}]
		When $\vect{x} = \vect{0}$, then 
		$\vect{b} = \abs{\mathcal{F}\vect{x}} = \vect{0}$, $\text{Aut}_p(\vect{x}) = \vect{0}$ and $\tilde{\vect{b}} = \vect{b} + \vect{\eta} = \vect{\eta}$. 
		\begin{align*}
		&\norm{\mathcal{F}^{-1}(\tilde{\vect{b}} \odot \tilde{\vect{b}}) - \text{Aut}_p(\vect{x})}_{\infty} = \norm{\mathcal{F}^{-1}(\vect{\eta} \odot \vect{\eta})}_{\infty} \leq \norm{\mathcal{F}^{-1}(\vect{\eta} \odot \vect{\eta})}_2 \\
		=& \dfrac{1}{\sqrt{N}}\norm{\vect{\eta} \odot \vect{\eta}}_2
		\leq \norm{\vect{\eta} \odot \vect{\eta}}_{\infty} 
		\leq \norm{\vect{\eta}}_{\infty}^2 < \dfrac{1}{64N^2} < \dfrac{1}{2}.
		\end{align*}
		When $\vect{x} \neq \vect{0}$, note $\norm{\vect x}_0 \leq N$ and $\norm{\vect{\eta}}_{\infty} < \dfrac{1}{8N} \leq \dfrac{1}{8\norm{\vect x}_ 0}$. 
		The rest is straightforward from Proposition~\ref{cor: denoising related to sparsity}.
	\end{proof}

	\begin{proof}[Proof of Corollary~\ref{cor: SNR}]
		The inequality $$\text{SNR}_{\text{dB}} > 10\log_{10}(64) + 30\log_{10}\norm{\vect x}_0,$$ is equivalent to $\dfrac{\norm{\vect x}_2^2}{\norm{\vect \eta}^2_2} > 64 \norm{\vect x}_0^3$. Since $\vect x \in \{0, 1\}^N$, $\norm{\vect x}_2^2 = \norm{\vect x}_0$. Thus, we have $\norm{\vect \eta}_{\infty}^2 \leq \norm{\vect \eta}_2^2 < \dfrac{1}{64\norm{\vect x}^2_0}$. 
	\end{proof}

	\bibliographystyle{vancouver}
	\bibliography{phaseRetrieval}
	
\end{document}